\documentclass[11pt]{article}
\pdfoutput=1
\usepackage[margin=1in]{geometry}
\usepackage{amsmath, amsthm, amsfonts, amssymb, epsfig, subfig}
\usepackage{tikz}
\usepackage{subfig}
\usetikzlibrary{arrows,shapes}

\newcommand{\set}[1]{\{{#1}\}}

\newcommand{\half}{\ensuremath{\frac{1}{2}}}

\newcommand{\cM}{\mathcal{M}}

\newcommand{\I}{\mathcal{I}}

\newcommand{\prob}{\mathcal{P}}

\newcommand{\m}{\cM}

\newcommand{\mnn}{\cM_{{nn}}}
\newcommand{\mt}{\cM_{{T}}}
\newcommand{\minv}{\cM_{{inv}}}

\newcommand{\mtree}{\cM_{tree}(T)}

\newcommand{\comment}[1]{}
\usepackage{ifthen} 
\newboolean{includefigs} 
\setboolean{includefigs}{true} 
\newcommand{\condcomment}[2]{\ifthenelse{#1}{#2}{}}
\theoremstyle{plain}
\newtheorem{theorem}{Theorem}[section]

\newtheorem{definition}{Definition}[section]
\newtheorem{lemma}[theorem]{Lemma}

\newtheorem{Remark}[theorem]{Remark}

\usepackage[bookmarks,bookmarksnumbered]{hyperref}

\title{Mixing Times of Self-Organizing Lists and Biased Permutations}

\author{
  Prateek Bhakta
    \thanks{College of Computing, Georgia Institute of
      Technology, Atlanta, GA 30332-0765.  Supported in part by NSF
      CCF-0830367 and a Georgia Institute of Technology ARC Fellowship. }
  \and Sarah Miracle
    \thanks{College of Computing, Georgia Institute of Technology, Atlanta, GA
      30332-0765.  Supported in part by a DOE Office of Science Graduate
      Fellowship, NSF CCF-0830367 and a ARCS Scholar Award.}
  \and Dana Randall
    \thanks{College of Computing, Georgia Institute of Technology, Atlanta, GA
      30332-0765.  Supported in part by NSF CCF-0830367 and CCF-0910584.}
  \and Amanda Pascoe Streib
    \thanks{School of Mathematics, Georgia Institute of Technology, Atlanta, GA
      30332-0280. Supported in part by the National Physical Sciences
      Consortium Fellowship and NSF CCF-0910584.}}

\begin{document}
\date{}
\maketitle

\thispagestyle{empty}  
\begin{abstract}
Sampling permutations from $S_n$ is a fundamental problem from probability
theory.  The nearest neighbor transposition chain ${\cal{M}}_{nn}$ is known to
converge in time $\Theta(n^3 \log n)$ in the uniform case \cite{wilson} and
time $\Theta(n^2)$ in the constant bias case, in which we put adjacent
elements in order with probability $p \neq 1/2$ and out of order with
probability $1-p$ \cite{BBHM05}.   Here we consider the variable bias case
where the probability of putting an adjacent pair of elements in order depends
on the two elements, and we put adjacent elements $x<y$ in order with
probability $p_{x,y}$ and out of order with probability $1-p_{x,y}$.  The
problem of bounding the mixing rate of ${\cal{M}}_{nn}$ was posed by Fill
\cite{F03b, F03a} and was motivated by the Move-Ahead-One self-organizing list
update algorithm.  It was conjectured that the chain would always be rapidly
mixing if $1/2 \leq p_{x,y} \leq 1$ for all $x < y$, but this was only known
in the case of constant bias or when $p_{x,y}$ is equal to 1/2 or 1, a case
that corresponds to sampling linear extensions of a partial order.  We prove
the chain is rapidly mixing for two classes: ``Choose Your Weapon,'' where we
are given $r_1, \dots, r_{n-1}$ with $r_i \geq 1/2$ and $p_{x,y}=r_x$ for all
$x<y$ (so the dominant player chooses the game, thus fixing his or her
probability of winning), and ``League Hierarchies,'' where there are two
leagues and players from the A-league have a fixed probability of beating
players from the B-league, players within each league are similarly divided
into sub-leagues with a possibly different fixed probability, and so forth
recursively.  Both of these classes include permutations with constant bias as
a special case.  Moreover, we also prove that the most general conjecture is
false.  We do so by constructing a counterexample where $1/2 \leq p_{x,y} \leq
1$ for all $x< y$, but for which the nearest neighbor transposition chain
requires exponential time to converge.  
\end{abstract}

\newpage
\setcounter{page}{1}
\section{Introduction}
Sampling from the permutation group $S_n$ is one of the most fundamental
problems in probability theory.  A natural Markov chain that has been studied
extensively is a symmetric chain that makes nearest neighbor transpositions,
${\cal{M}}_{nn}.$  After a series of papers  \cite{diasha, DSC93b} 
Wilson \cite{wilson} showed a tight bound of $\Theta(n^3 \log n)$ on the
mixing time, with upper and lower bounds within a factor of two.  Subsequently
Benjamini et al. \cite{BBHM05} considered a biased version of this Markov
chain where we select a pair of adjacent elements at random and put them in
order with probability $p > 1/2$ and out of order with probability $1-p$.
They relate this biased shuffling Markov chain to a chain on an asymmetric
simple exclusion process (ASEP) and showed that they both converge in $\Theta
(n^2)$ time. These bounds were matched by Greenberg et al. \cite{GPR09} who
also generalized the result on ASEPs to sampling biased surfaces in two and
higher dimensions in optimal $\Theta(n^d)$ time.  

In this paper we consider a generalization where we are always at least as
likely to put a pair of adjacent elements in increasing order as out of order,
but where the bias can vary depending on the values of the two elements.  More
precisely, we are given input parameters ${\bf P} = \{p_{i,j}\}$ for all $1
\leq i, j \leq n$.  The Markov chain ${\cal{M}}_{nn}$ iteratively chooses a
pair of adjacent elements uniformly, and if they are $i$ and $j$ we put $i$
ahead of $j$ with probability $p_{i,j}$ and we put $j$ ahead of $i$ with
probability $p_{j,i} = 1-p_{i,j}$.  We are interested in understanding whether
$\mnn$ is efficient in this generalized context.  We call the case where $1/2
\leq p_{i,j} \leq 1$ for all $i< j$ \emph{positively biased}.  In this case,
the fully ordered permutation $1, 2, \dots, n$ is at least as likely in
stationarity as every other permutation.   It is not difficult to see that
$\mnn$ can take exponential time without this condition.  

The problem of bounding the mixing rate of $\mnn$ in the variable bias setting
was raised by Jim Fill \cite{F03b, F03a} who considered it in the context of
the Move-Ahead-One (MA1) self-organizing list update algorithm.  In the MA1
protocol, elements are chosen according to some underlying distribution and
they move up by one in a linked list after each request is serviced, if
possible.  Thus, the most frequently requested elements will move to the front
of the list and will eventually require less access time. If we consider a
pair of adjacent elements $i$ and $j$, the probability of performing a
transposition that moves $i$ ahead of $j$ is proportional to $i$'s request
frequency, and similarly the probability of moving $j$ ahead of $i$ is
proportional to $j$'s frequency, so the transposition rates vary depending on
$i$ and $j$ and we are always more likely to put things in order (according to
their request frequencies) than out of order.  Fill conjectured that when the
transposition probabilities ${\bf P}$ also satisfy a monotonicity condition
whereby $p_{i,j} \leq p_{i, j+1}$ and $p_{i,j} \geq p_{i+1,j}$ for all $1 \leq
i < j \leq n$, then the chain is always rapidly mixing.  In fact, he
conjectured that the spectral gap is always minimized when $p_{i,j}=1/2$ for
all $i,j$, a problem he refers to as the ``gap problem.''  He verified that
the conjecture is true for $n=4$ and gave experimental evidence for slightly
larger $n$.  

Although Fill posed the gap problem in a widely circulated manuscript ten
years ago, there has been very little progress toward solving it.  For general
$n,$ the chain has only been shown to be rapidly mixing in two settings. The
first is the constant bias case for which Benjamini et al. \cite{BBHM05}
showed a mixing time of $\theta(n^2)$ when $p_{i,j}=p > 1/2$ for all $i<j.$ 
The second case has all of the $p_{i,j}$ with $i < j$ equal to 1/2 or 1; in
this context the nearest neighbor chain ${\cal{M}}_{nn}$ samples linear
extensions of a partial order and was shown by Bubley and Dyer \cite{BD98} to
mix in $O(n^3 \log n)$ time.

\vskip.1in
\underbar{\bf Our results:} \ 
In this paper we show that the Markov chain ${\cal{M}}_{nn}$  is always
rapidly mixing for two significantly larger classes of inputs which we call
``Choose Your Weapon'' and ``League Hierarchies.''  In the Choose Your Weapon
class we are given a set of input parameters $r_1, \dots, r_{n-1}$
representing each player's ability to win a duel with his or her weapon of
choice.  When a pair of neighboring players are chosen to compete, the
dominant player gets to choose the weapon, thus determining his or her
probability of winning the match.  In other words, we set $p_{i,j} = r_i$ when
$i< j$. We show that the nearest neighbor transposition chain ${\cal{M}}_{nn}$
is rapidly mixing for any choice of $\{r_i\}$.   In the League Hierarchy class
we are given input parameters $q_1, \dots, q_{n-1}$ along with a binary tree
$T$ whose internal vertices are labeled with the $q_i$ and whose leaves are
labeled with the elements $1, \dots, n$.  We think of the leaves descending
from the left branch of the root as the A-league and the right branch as the
B-league, and whenever players from the two leagues are matched up, the player
from the A-league has an advantage indicated by the probability associated
with the root.  Likewise, within the A-league we have Tier-1 and Tier-2
players, and the probability that Tier-1 players win matches against Tier-2
players is determined by the probability at the root of that subtree.  Thus,
$p_{i,j} = q_{i \wedge j}$ for all $i < j$, where $i \wedge j$ is the lowest
common ancestor of the leaves labelled $i$ and $j$.  We show that a related
chain  including additional transpositions is rapidly mixing for any choice of
$\set{q_i}$, and that ${\cal{M}}_{nn}$ is also if the $\set{q_i}$ additionally
satisfies ``weak monotonicity'' (i.e., $p_{i,j} \leq p_{i,j+1}$ if $j > i$).  
We note that both of these classes are generalizations of the uniform bias
setting, which can be seen by taking all of the $r_i$ or $q_i$ to be constant.  

In addition, we disprove the most general form of the conjecture by
constructing a set ${\bf P}$ for which the chain requires exponential time,
even in the positive bias case where $p_{i,j} > 1/2$ for all $i<j$.  Our
example is motivated by models in statistical physics that exhibit a phase
transition arising from a ``disordered phase'' of high entropy and low energy,
an ``ordered phase'' of high energy and low entropy, and a bad cut separating
them that is both low energy and entropy.  This example does not satisfy the
monotonicity condition of Fill, but does give insight into why bounding the
mixing rate of the chain in more general settings has proven quite
challenging.

\vskip.1in
\underbar{\bf Techniques:} 
For the positive results, our strategy is to use various combinatorial
representations of permutations and interpret the moves of ${\cal{M}}_{nn}$ in
these new settings.  In each case there is a natural Markov chain in the new
setting including additional moves (also transpositions) that can be analyzed
using simple arguments.  We then reinterpret the new moves in terms of the
original permutations so that we can deduce bounds on the mixing rate of the
nearest neighbor transposition chain as well.  In each case the new Markov
chain consists of a family of transpositions and are themselves interesting in
the context of generating random permutations.

For the Choose Your Weapon class, we map permutations to {\it Inversion
Tables} \cite{Knuth, Turrini} that, for each element $i$, record how many
elements $j > i$ come before $i$ in the permutation.  We consider a  Markov
chain ${\cal{M}}_{inv}$ that simply increments or decrements a single element
of the inversion table in each step; using the bijection with permutations
this corresponds to adding additional transpositions of elements that are not
necessarily nearest neighbors to the Markov chain ${\cal{M}}_{nn}$.
Remarkably, this allows ${\cal{M}}_{inv}$ to decompose into a product of
simple one-dimensional random walks and bounding the convergence time is very
straightforward.  Finally, we use comparison techniques \cite{DSC93, RT98} to
bound the mixing time of the nearest neighbor chain ${\cal{M}}_{nn}$ for all
choices of inputs $r_1, \dots, r_{n-1}$.  This approach also gives new, far
simpler proof of fast mixing in the case of uniform bias.

For the League Hierarchy class, we introduce a new combinatorial
representation of the permutation that associates a bit string $b_v$ to each
node $v$ of a binary tree with $n$ leaves.  Specifically, $b_v \in \{L,
R\}^{\ell_v}$ where $\ell_v$ is the number of leaves in $t_v$, the subtree
rooted at $v$, and for each element $i$ of the sub-permutation corresponding
to the leaves of $t_v$, $b_v(i)$ records whether $i$ lies under  the left or
the right branch of $v$.  The set of these bit strings is in bijection with
the permutations.  We consider a chain ${\cal{M}}_{tree}$ that allows
transpositions exactly when they correspond to a nearest neighbor
transposition in exactly one of the bit strings.  Thus, the mixing time of
${\cal{M}}_{tree}$ decomposes into a product of $n-1$ ASEP chains and we can
conclude that the chain ${\cal{M}}_{tree}$ is rapidly mixing using results in
the uniform bias case \cite{BBHM05, GPR09}.  Again, we use comparison
techniques to conclude that the nearest neighbor chain is also rapidly mixing
when we have weak monotonicity, although ${\cal{M}}_{tree}$ which simply
allows additional transpositions is always rapidly mixing.  

For the negative result showing slow mixing, the choice of ${\bf P}$ was
motivated by a related question arising in the context of biased staircase
walks \cite{GPR09}.  In that context, we are sampling ASEP configurations with
$n$ zeros and $n$ ones, which map bijectively onto walks on the Cartesian
lattice from $(0, n)$ to $(n, 0)$ that always go to the right or down.  The
probability of each walk $w$ is proportional to $\Pi_{xy < w} \lambda_{xy}$,
where the bias $\lambda_{xy}\geq 1/2$ is assigned to the square at $(x, y)$
and $xy < w$ whenever the square at $(x, y)$ lies underneath the walk $w$.  We
show that there are settings of the $\{\lambda_{xy}\}$ which cause the chain
to be slowly mixing  from any starting configuration (or walk).  In
particular, we show that at stationarity the most likely configurations will
be concentrated near the diagonal from $(0,n)$ to $(n,0)$ (the high entropy,
low energy states) or they will extend close to the point $(n,n)$ (the high
energy, low entropy states) but it will be unlikely to move between these sets
of states because there is a bottleneck that has both low energy and low
entropy.  
Finally, we use the reduction from biased permutations to biased lattice paths
to produce a positively biased set of probabilities {\bf P} for which
${\cal{M}}_{nn}$ also requires exponential time to mix from any starting
configuration.

\section{Preliminaries}

We begin by formalizing our model. Let $\Omega = S_n$ be the set of all
permutations $\sigma=(\sigma(1),\sigma(2),\ldots, \sigma(n))$ of $n$ integers.
We consider Markov chains on $\Omega$ whose transitions transpose two elements
of the permutation.  A permutation $\sigma$ is represented as a list of
elements, $\sigma(1),\sigma(2),\ldots, \sigma(n)$.  We are also given a set
${\bf P}$, consisting of $p_{i,j}\in [0,1]$ for each $1\leq i\neq j\leq n$,
where for any $i<j$, $p_{i,j} \geq 1/2$ and $p_{j,i}= 1 - p_{i,j}$.
The Markov chain $\mnn$ will sample from $\Omega$ using ${\bf P}$.

\vspace{.1in} \noindent  {\bf The Nearest Neighbor Markov chain $\mnn$ } 

\vspace{.05in} \noindent {\tt Starting at any permutation $\sigma_0$, iterate
the following:

\begin{itemize}
  \item At time $t$, select an index $i\in [n-1]$ uniformly at random (u.a.r).
    \begin{itemize}
      \item Swap the elements $\sigma_t(i),\sigma_t(i+1)$ with probability
        $p_{\sigma_t(i+1),\sigma_t(i)}$ to obtain $\sigma_{t+1}$.
      \item Do nothing with probability $p_{\sigma_t(i),\sigma_t(i+1)}$ so
        that $\sigma_{t+1} = \sigma_t$.
    \end{itemize}
\end{itemize}
}
\noindent The Markov chain $\mnn$ connects the state space, since every
permutation $\sigma$ can move to the ordered permutation $(1,2,\ldots, n)$
(and back) using the bubble sort algorithm.
Since $\mnn$ is also aperiodic, this implies that $\mnn$ is ergodic.   
For an ergodic Markov chain with transition probabilities $\prob$, if some
assignment of probabilities $\pi$ satisfies the
\emph{detailed balance condition} $\pi(\sigma)\prob(\sigma,\tau) =
\pi(\tau)\prob(\tau,\sigma)$ for every $\sigma,\tau \in \Omega$, then $\pi$ is
the stationary distribution of the Markov chain \cite{LPW06}.
It is easy to see that for $\mnn$, the distribution
$\pi(\sigma) =\prod_{(i < j)} p_{\sigma(i),\sigma(j)}/Z$,
where $Z$ is the normalizing constant $\sum_{\sigma \in \Omega}
\prod_{(i<j)} p_{\sigma(i),\sigma(j)}$, satisfies detailed balance,
and is thus the stationary distribution.  

The Markov chain $\mt$ can make any transposition at each
step, while maintaining the stationary distribution $\pi$. The
transition probabilities of $\mt$ can be quite complicated, since swapping
two distant elements in the permutation consists of many
transitions of $\mnn$, each with different probabilities. 
In the following sections, we will introduce two other Markov chains whose
transitions are a subset of those of $\mt$ for which we can describe the
transition probabilities succinctly.

\subsection{Convergence rates of Markov chains}
Next, we present some background on Markov chains. 
The \emph{total variation distance} between the stationary
distribution $\pi$ and the distribution of the Markov Chain at time $t$ is
$\|\prob^t,\pi\| _{tv} = \max_{x\in\Omega}\frac{1}{2}\sum_{y\in\Omega}
  |\prob^t(x,y)-\pi(y)|,$
where $\prob^t(x,y)$ is the $t$-step transition probability.  The efficiency
of a Markov chain $\m$ is often measured by its \emph{mixing time}
$\tau(\epsilon)$.  For all $\epsilon>0$, we define $\tau(\epsilon)=\min \{t:
\|\prob^{t'},\pi \|_{tv}\leq \epsilon, \forall t' \geq t\}.$  We say that a
Markov chain is \emph{rapidly mixing} if there exists a polynomial $p$ such
that $\tau_{\epsilon} = O(p(n,\log(\varepsilon^{-1})))$ where $n$ is the size
of each configuration in $\Omega$.

In Section~\ref{inversionSec}, we will use a standard technique called
\emph{coupling}.  A coupling is a Markov chain $(X_t,Y_t)_{t=0}^\infty$ on
$\Omega\times \Omega$ such that each of the processes $X_t$ and $Y_t$ is a
faithful coupling of $\m$, and if $X_t=Y_t$, then $X_{t+1}=Y_{t+1}$.  Given
such a coupling, define the \emph{coupling time} $T$ as follows:
$$T=\max_{x,y}E[\min\{t: X_t = Y_t | X_0=x,Y_0=y\}].$$ Then the following
theorem (see, e.g.~\cite{ald}) relates the coupling time and the mixing time.
\begin{theorem}\label{coupling}
  $\tau(\epsilon)\leq Te\lceil \ln\epsilon^{-1}\rceil.$
\end{theorem}

In each of Sections~\ref{inversionSec} and~\ref{treeSec}, we introduce new
Markov chains to sample from the same distribution as $\mnn$.  In order to
obtain bounds on the mixing time of $\mnn$, we will compare $\mnn$ with these
auxiliary chains in Section~\ref{comparisonSec}.  If $P$ and $P'$ are the
transition matrices of two reversible Markov chains on the same state space
$\Omega$ with the same stationary distribution $\pi$, the comparison method
(see \cite{DSC93} and \cite{RT98}) allows us to relate the mixing times of
these two chains.  Let $E(P) = \{(\sigma,\beta): P(\sigma,\beta) > 0\}$ and
$E(P') = \{(\sigma,\beta): P'(\sigma,\beta) > 0\}$ denote the sets of edges of
the two graphs, viewed as directed graphs.  For each $\sigma,\beta$ with
$P'(\sigma,\beta)>0$, define a path $\gamma_{\sigma\beta}$ using a sequence of
states $\sigma=\sigma_0,\sigma_1,\cdots,\sigma_k = \beta$ with
$P(\sigma_i,\sigma_{i+1})>0$, and let $|\gamma_{\sigma\beta}|$ denote the
length of the path.
Let $\Gamma(\upsilon,\omega) = \{(\sigma,\beta) \in E(P'): (\upsilon,\omega)
\in \gamma_{\sigma\beta}\}$ be the set of paths that use the transition
$(\upsilon,\omega)$ of $P$.  Finally, let
$\pi_*=\min_{\rho\in\Omega}\pi(\rho)$ and define $$A = \max_{(\upsilon,\omega)
\in
E(P)}\frac{1}{\pi(\upsilon)P(\upsilon,\omega)}\sum_{\Gamma(\upsilon,\omega)}|\gamma_{\sigma\beta}|\pi(\sigma)P'(\sigma,\beta)
.$$

\noindent The following formulation of the comparison method is due to Randall
and Tetali~\cite{RT98}.

\begin{theorem}
  \label{Comparison}
  With the above notation, for $0<\epsilon<1$,
  we have $\tau(\epsilon) \leq \frac{4\log(1/(\epsilon
  \pi_*))}{\log(1/2\epsilon)}A\tau'(\epsilon).$
\end{theorem}

\section{Choose Your Weapon}
\label{inversionSec}
In the Choose Your Weapon class, we are given $1/2 \leq r_1,r_2,\ldots,
r_{n-1} \leq 1$ and a set ${\bf P}$ satisfying $p_{i,j}=r_i$, if $i<j$ and
$p_{i,j}=1-p_{j,i}$ if $j<i$.  We show that a new Markov chain $\minv$ is
rapidly mixing under these conditions, which will imply that $\mnn$ and $\mt$
are as well, as we show in Section~\ref{comparisonSec}.  The Markov chain
$\minv$ acts on the \emph{inversion table} of the permutation~\cite{Knuth,
Turrini}, which has an entry for each $i\in [n]$ counting the number of
inversions involving $i$; that is, the number of values $j>i$ where $j$ comes
before $i$ in the permutation (see Figure~\ref{inversiontable}).  It is easy
to see that the $i$th element of the inversion table is an integer between $0$
and $n-i$.  In fact, the function $I$ is a bijection between the set of
permutations and the set $\I$ of all possible inversion tables (all sequences
$X=(x_1,x_2,\ldots, x_n)$ where $0\leq x_i\leq n-i$ for all $i\in [n]$).  To
see this, we will construct a permutation from any inversion table $X\in I$.
Place the element $1$ in the $(x_1+1)$st position of the permutation.  Next,
there are $n-1$ slots remaining.  Among these, place the element $2$ in the
$(x_2+1)$st position remaining (ignoring the slot already filled by 1).
Continuing, after placing $i-1$ elements into the permutation, there are
$n-i+1$ slots remaining, and we place the element $i$ into the $(x_i+1)$st
position among the remaining slots.  This proves that $I$ is a bijection from
$S_n$ to $\I$.

\begin{figure}
  \center{$ \begin{array}{cccccccccc}
    \sigma &=& 8 & 1 & 5 & 3 &7 & 4 &6 &2 \\ I(\sigma)&=&1 & 7 & 2 & 3 & 1 & 2
    &1 &0 \end{array} $
  } \caption{The inversion table for a permutation.}
  \label{inversiontable}
\end{figure}

Given this bijection, a natural algorithm for sampling permutations is to
perform the following local Markov chain on inversion tables: select a
position $i\in [n]$ and attempt to either add one or subtract one from $x_i$,
according to the appropriate probabilities.  In terms of permutations, this
amounts to adding or removing an inversion involving $i$ without affecting the
number of inversions involving any other integer, and is achieved by swapping
the element $i$ with an element $j>i$ such that every element in between is
smaller than both $i$ and $j$.  If $i$ moves ahead of $j$, this move happens
with probability $p_{i,j}$ because for each $k$ that $i$ and $j$ are swapped
past, $k<i,j$, so $p_{k,i}=r_k = p_{k,j}$ (since each of these depend only on
$k$) so the net effect on the distribution is neutral, and the detailed
balance condition ensures that $\pi$ is the correct stationary distribution.
Formally, the Markov chain is defined as follows.

\vspace{.1in} \noindent  {\bf The Inversion Markov chain $\minv$ } 

\vspace{.05in} \noindent {\tt Starting at any permutation $\sigma_0$, iterate
the following:

\begin{itemize}
  \item  Select an element $i\in [n]$ with probability $(n-i)/\binom{n}{2}$
    and a bit $b\in \{-1,+1\}$.
    \begin{itemize}
      \item If $b=+1$, let $j$ be the first element after element $i$ in
        $\sigma_t$ such that $j>i$.  With prob. $p_{j,i}/2=(1-r_i)/2$, obtain
        $\sigma_{t+1}$ from $\sigma_t$ by swapping $i$ and $j$.  
      \item If $b=-1$, let $j$ be the last element before element $i$ in
        $\sigma_t$ such that $j>i$.  With prob. $p_{i,j}/2=r_i/2$, obtain
        $\sigma_{t+1}$ from $\sigma_t$ by swapping $i$ and $j$.  
    \end{itemize}
  \item With prob. $1/2$, $\sigma_{t+1}=\sigma_t$.
\end{itemize}
}
\noindent 
This Markov chain contains the moves of $\mnn$ (and therefore also connects
the state space).  Although elements can jump across several elements, it is
still fairly local compared with the general transposition chain $\mt$ which
has $\binom{n}{2}$ choices at every step, since $\minv$ has at most $2n$.

The Markov chain $\minv$ is essentially a product of $n$ independent
one-dimensional processes.  The $i$th process is just a random walk bounded
between 0 and $n-i$, which moves up with probability $1-r_i$ and down with
probability $r_i$; hence its mixing time is $O(n^2)$, unless $r_i$ is bounded
away from $1/2$, in which case its mixing time is $O(n)$.   However, each
process is slowed down by a factor of $n$ since we only update one process at
each step.  To make this argument formal, we will use
Theorem~\ref{crossproductthm}, which bounds the mixing time of a product of
independent Markov chains and whose elementary proof is deferred to
Section~\ref{productsection}.  

\begin{theorem}\label{BiasedInversions} Let $1/2\leq r_1, r_2,\ldots,
  r_{n-1}< 1$ be constants, and let $r_{max}=\max_i r_i$.  Assume that
  $p_{i,j}=r_{\min\{i,j\}}.$  
  \begin{enumerate}
    \item If each $r_i>1/2$ then the mixing time of $\minv$ on biased
      permutations with these $p_{i,j}$ values is $O(n^2\ln(n/\epsilon))$. 
    \item Otherwise,  the mixing time of $\minv$  is $O(n^3\ln(n/\epsilon))$. 
  \end{enumerate}
\end{theorem}

To prove this theorem, we need to analyze the one-dimensional process
$\m(r,k)$, bounded between $0$ and $k$, which chooses to move up with
probability $r\geq 1/2$ and down with probability $1-r$ at each step, if
possible.  This simple random walk is well-studied; we include the proof for
completeness. 

\begin{lemma}\label{BiasedInversionsOned}
  Let $1/2\leq r\leq 1$ be constant.  Then the Markov chain $\m(r,k)$ has
  mixing time 
  \begin{enumerate}
    \item $\tau(\epsilon)=O(k\ln\epsilon^{-1})$ if $r$ is a constant bigger than
      $1/2$, and
    \item $\tau(\epsilon)=O(k^2\ln\epsilon^{-1})$ if $r=1/2$.
  \end{enumerate}
\end{lemma}
\begin{proof}
  We use a variation on coupling.  We use the trivial coupling, which chooses
  to move the same direction in each Markov chain.  Notice that the Markov
  chain $\m(r,k)$ is monotone with respect to this coupling, in the sense that
  if $X_t$ is below $Y_t$, then it will remain so until $X_{t'}=Y_{t'}$.  Thus
  the time until the chains couple is bounded by the time it takes for a
  process $Z_t$, where $Z_0=0$, to reach height $k$.  However, $Z_t$ is just a
  biased random walk bounded between $0$ and $k$.   First, we notice that
  $Z_t$ is non-decreasing in expectation; that is, for all $t>0$,
  $E[Z_{t+1}-Z_{t}]\geq 0$:
  \begin{eqnarray*}
    E[Z_{t+1}-Z_{t}]&=& r - (1-r) = 2r-1\geq 0.
  \end{eqnarray*}
 
  Consider the case that $r>1/2$.  Define $W(t)=k-Z(t) +(2r-1)t$.  Examining
  the expected difference between $W(t)$ and $W(t+1)$, we see $$E[W(t+1)-W(t)]
  = E[-Z(t+1)+2r-1+Z(t)] = 0.$$ Also, since the differences $W(t+1) - W(t)$ are
  bounded, $\{W (t)\}$ is a martingale.  The time $T=\min\{t: Z_t=0\}$ is a
  stopping time for the process $W(t)$, so we may apply the Optional Stopping
  Theorem for martingales to deduce that $$E[W(T)]=W(0)=k.$$ However, since
  $$E[W(T)]=E[k-Z(T) + (2r-1)T]= (2r-1)E[T],$$ it follows that $E[T]=k/(2r-1).$
  Recall from Theorem~\ref{coupling} that 
  $$\tau(\epsilon)=O(T\ln\epsilon^{-1}) = O(k/(2r-1)\ln\epsilon^{-1})=O(k
  \ln\epsilon^{-1}).$$

  Suppose now that $r=1/2$. 
  This case is similar, and follows from Lemma~6 of~\cite{lrs}.   Notice
  $E[(Z(t+1)- Z(t))^2] = r +(1-r) = 1.$
  Therefore
  $E[T]\leq k (2k-k)/1= k^2.$  Hence 
  $\tau(\epsilon)= O(k^2 \ln\epsilon^{-1}).$

\end{proof}

Finally, we can use these bounds to prove Theorem~\ref{BiasedInversions}.  

\vspace{.2in}
\noindent {\em Proof of Theorem~\ref{BiasedInversions}.}  \ 
The $i$th process is chosen with probability $(n-i)/(2\binom{n}{2})$.
Therefore, by Theorem~\ref{crossproductthm}, the mixing time of $\minv$
satisfies 
$$\tau(\epsilon)\leq \frac{\binom{n}{2}}{n-i} (n-i)\ln(2n/\epsilon) = \binom{n}{2}\ln(2n/\epsilon)=O(n^2\ln(n/\epsilon))$$
when each $r_i$ is bounded away from $1/2$. Otherwise, 
$$\tau(\epsilon)\leq \frac{\binom{n}{2}}{n-i} (n-i)^2\ln(2n/\epsilon) =O(n^3\ln(n/\epsilon)).$$
\qed

\begin{Remark}
The same proof also applies to the case where the probability of swapping $i$
and $j$ depends on the object with lower rank (i.e., we are given $r_2, \dots
r_n$ and we let $p_{i,j}=r_j$ for all $i<j$).  This case is related to a
variant of the MA1 list update algorithm, where if a record is requested, we
try to move the associated record $x$ ahead of its immediate predecessor in
the list, if it exists.  If it has higher rank than its predecessor, then it
always succeeds, while if its rank is lower we move it ahead with probability
$f_x =r_x/ (1 + r_x)\leq 1$.  
\end{Remark}

\section{League Hierarchy}
\label{treeSec}

In this section, we turn to a second class of ${\bf P}$ that have what we call
\emph{league structure}.  Let $T$ be a proper rooted binary tree with $n$ leaf
nodes, labeled $1, \ldots, n$ in sorted order. Each non-leaf node $v$ of this
tree is labeled with a value $\half \leq q_v \leq 1$. For $i,j\in [n]$, let $i
\vee j$ be the lowest common ancestor of the leaves labeled $i$ and $j$.  We
say that ${\bf P}$ \emph{has league structure $T$} if for all $i<j$, $p_{i,j}
= q_{i \vee_T j}$ and $p_{j,i}=1-p_{i,j}$. For example,
Figure~\ref{figtreepij} shows a set ${\bf P}$ such that $p_{14} = .8$, $p_{49}
= .9$, and $p_{58} = .7$. 

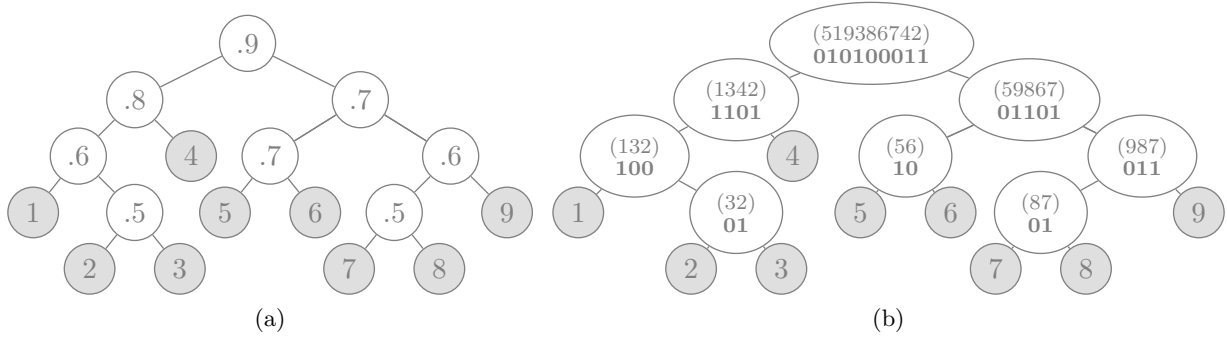
\begin{figure}[h]
  \centering
    \subfloat[]{
    \begin{tikzpicture}[scale=.15]
      \tikzstyle{node}=[circle,edge=black!100]
      \tikzstyle{fnode}=[circle,edge=black!100,fill=lightgray!50]
      \tikzstyle{edge}=[draw,line width = .5pt,gray]
      \node[node] (v) at (30,20) {$.9$};
      \node[node] (v1) at (20,15) {$.8$};
      \node[node] (v2) at (40,15) {$.7$};
      \node[node] (v11) at (15,10) {$.6$};
      \node[fnode] (v12) at (25,10) {$4$};
      \node[node] (v21) at (32,10) {$.7$};
      \node[node] (v22) at (48,10) {$.6$};
      \node[fnode] (v111) at (11,5) {$1$};
      \node[node] (v112) at (20,5) {$.5$};
      \node[fnode] (v211) at (28,5) {$5$};
      \node[fnode] (v212) at (36,5) {$6$};
      \node[node] (v221) at (43,5) {$.5$};
      \node[fnode] (v222) at (53,5) {$9$};
      \node[fnode] (v1121) at (16,0) {$2$};
      \node[fnode] (v1122) at (24,0) {$3$};
      \node[fnode] (v2211) at (39,0) {$7$};
      \node[fnode] (v2212) at (47,0) {$8$};
      \foreach \t in {v,v1,v2,v2,v11,v112,v21,v22,v221} {
        \draw[edge] (\t) -- (\t1);
        \draw[edge] (\t) -- (\t2);
      }
    \end{tikzpicture}
      \label{figtreepij}
    }
    \subfloat[]{
    \begin{tikzpicture}[scale=.15]
      \tikzstyle{node}=[ellipse,edge=black!100]
      \tikzstyle{fnode}=[circle,edge=black!100,fill=lightgray!50]
      \tikzstyle{edge}=[draw,line width = .5pt,gray]
      \node[node] (v) at (30,20) {$(519386742) \atop \bf 010100011$};
      \node[node] (v1) at (18,15) {$(1342) \atop \bf 1101$};
      \node[node] (v2) at (44,15) {$(59867) \atop \bf 01101$};
      \node[node] (v11) at (9,10) {$(132) \atop \bf 100$};
      \node[fnode] (v12) at (23,10) {$4$};
      \node[node] (v21) at (33,10) {$(56) \atop \bf 10$};
      \node[node] (v22) at (54,10) {$(987) \atop \bf 011$};
      \node[fnode] (v111) at (4,5) {$1$};
      \node[node] (v112) at (18,5) {$(32) \atop \bf 01$};
      \node[fnode] (v211) at (29,5) {$5$};
      \node[fnode] (v212) at (37,5) {$6$};
      \node[node] (v221) at (45,5) {$(87) \atop \bf 01$};
      \node[fnode] (v222) at (59,5) {$9$};
      \node[fnode] (v1121) at (14,0) {$2$};
      \node[fnode] (v1122) at (22,0) {$3$};
      \node[fnode] (v2211) at (41,0) {$7$};
      \node[fnode] (v2212) at (49,0) {$8$};
      \foreach \t in {v,v1,v2,v2,v11,v112,v21,v22,v221} {
        \draw[edge] (\t) -- (\t1);
        \draw[edge] (\t) -- (\t2);
      }
    \end{tikzpicture}
    \label{figtreeenc}
    }
    \caption{
    \small
    A set ${\bf P}$ with league structure, and the corresponding
    tree-encoding of the permutation 519386742.
    }
    \label{figtree}
\end{figure}

When $T$ is a complete binary tree and $q_{v_1}=q_{v_2}$ for each $v_1$ and
$v_2$ on the same level of the tree, this is precisely the representation of
the winning probabilities for a tournament described in the introduction.  We
define the Markov chain $\mtree$ over permutations, given a set ${\bf P}$ with
league structure $T$.

\vspace{.1in}
\noindent {\bf The Markov chain $\mtree$ } 

\vspace{.05in}
\noindent {\tt Starting at any permutation $\sigma_0$, iterate the following:

\begin{itemize}
  \item Select distinct $a,b \in [n]$ u.a.r.
    Assume $a < b$.
  \item If every number between $a$ and $b$ in the permutation $\sigma_t$
    is not a descendant in $T$ of $a\vee_T b$, obtain $\sigma_{t+1}$ from
    $\sigma_t$ by placing $a,b$ in order with probability
    $p_{a,b}$, and out of order with probability $1 - p_{a,b}$, leaving all
    elements between them fixed.
  \item Otherwise, $\sigma_{t+1}=\sigma_t$.
\end{itemize}
}

First, we show that this Markov chain samples from the same distribution as
$\mnn$. Swapping arbitrary non-adjacent elements $a$ and $b$ could
potentially change the weight of the permutation dramatically.
However, for any element $c$ that is not a descendant in $T$ of
$a\vee_T b$, the relationship between $a$ and $c$ is the same as the
relationship between $b$ and $c$. Thus the league structure ensures that
swapping $a$ and $b$ only changes the weight by a multiplicative factor
of $\lambda_{a,b}=p_{a,b}/p_{b,a}$.  
\begin{lemma}
  \label{mtreedistribution}
  The Markov chain $\mtree$ has the same stationary distribution as $\mnn$.
\end{lemma}
\begin{proof}
  Let $\pi$ be the stationary distribution of $\mnn$, and let
  $(\sigma_1,\sigma_2)$ be a transition in $\mtree$. It suffices to show that
  the detailed balance condition holds for this transition with the stationary
  distribution $\pi$.  Recall that we may express $\pi(\sigma) ={\prod_{i,j
  \vert i <_{\sigma} j} p_{i,j}}/{Z}$ where $Z = \sum_{\sigma \in \Omega}
  \prod_{i,j \vert i <_{\sigma} j} p_{i,j}$.  The transition
  $(\sigma_1,\sigma_2)$ transposes some two elements $a<_{\sigma_1}b$, where
  every element between $a$ and $b$ in $\sigma_i$ is not a descendant of
  $a\vee b$ in $T$. Let $x_1, \ldots, x_k$ be those elements.  Thus, the path
  from $a$ or $b$ to $x_i$ in $T$ must pass through $a \vee b$ and go to
  another part of the tree. For every such element $x_i$, $a \vee x_i = (a
  \vee b) \vee x_i = b \vee x_i$.  From the observation, we see from the
  league structure that $p_{ax_i} = p_{bx_i}$ for every $x_i$ between $a$ and
  $b$.  In particular, $x_i$ is either greater than both $a$ and $b$ or less
  than both $a$ and $b$, since all integers $c$ such that $a<c<b$ are
  necessarily descendants of $a \vee b$.  Therefore,
  \[\frac{\pi(\sigma_1)}{\pi(\sigma_2)} = \frac{ p_{ab}\prod_i p_{ax_i}}
  {p_{ba} \prod_i p_{bx_i}} = \frac{p_{ab}}{p_{ba}}.\] This is exactly the
  ratio of the transition probabilities in $\mtree$, thus $\mtree$ also has
  stationary distribution $\pi$.
\end{proof}

The key to the proof that $\mtree$ is rapidly mixing is again to decompose the
chain into $n-1$ independent Markov chains, $\m_1,\m_2,\ldots, \m_{n-1}$, one
for each non-leaf node of the tree $T$. To this end, we introduce an alternate
representation of a permutation as a set of binary strings arranged like the
tree $T$.  For each non-leaf node $v$ in the tree $T$, let $L(v)$ be its left
descendants, and $R(v)$ be its right descendants. We now do the following:
Given the permutation $\sigma$, list each descendant $x$ of $v$ in the order
we encounter it in $\sigma$; these are parenthesized in
Figure~$\ref{figtreeenc}$.  Then for each listed element $x$, write a $1$ if
$x \in L(v)$ and a $0$ if $x \in R(v)$. This is the final binary encoding in
Figure~$\ref{figtreeenc}$.  We see that any $\sigma$ will lead to an
assignment of binary strings at each non-leaf node $v$ with $L(v)$ ones and
$R(v)$ zeroes.  Next we verify that this is a bijection between the set of
permutations and the set of assignments of such binary strings to the tree
$T$.  Given any such assignment of binary strings, we can recursively
reconstruct the permutation $\sigma$ as follows. For each leaf node $i$, let
its string be the string $``i"$.  For any node $n$ with binary string $b$,
determine the strings of its two children. Call these $s_1,s_0$.  Interleave
the elements of $s_1$ with $s_0$, choosing an element of $s_1$ for each $1$ in
$b$, and an element of $s_0$ for each $0$.  This yields a permutation
$\sigma$.
 
With this bijection, we first analyze $\mtree$'s behavior over tree
representations and later extend this analysis to permutations.  The Markov
chain $\mtree$, when proposing a swap of the elements $a$ and $b$, will only
attempt to swap them if $a,b$ correspond to some \emph{adjacent} 0 and 1 in
the string associated with $a \vee b$.  Swapping $a$ and $b$ does not affect
any other string, so each non-leaf node $v$ represents an independent
exclusion process with $L(v)$ ones and $R(v)$ zeroes.  These exclusion
processes have been well-studied~\cite{BD98,wilson,BBHM05,GPR09}.  We will use
the following bounds on the mixing times of the symmetric and asymmetric
simple exclusion processes.
  
\begin{theorem}
  \label{exclusion}
  Let $\m$ be the exclusion process with parameter $p$ on $k_1$ ones and
  $k_2$ zeroes, where $k=k_1+k_2$.  Then
  \begin{enumerate} 
    \item if $p=1/2$, $\tau(\epsilon)=O(k^3\log (k_1k_2/\epsilon))$.
      \cite{BD98, wilson}
    \item if $p>1/2$, then $\tau(\epsilon)=O(k(\min\{k_1,k_2\} + \log
      k)\log(\epsilon^{-1}))=O(k^2\log(\epsilon^{-1}))$. \cite{GPR09}
  \end{enumerate} 
\end{theorem}

\noindent The bounds in Theorem~\ref{exclusion} refer to the exclusion process
which selects a position at random and swaps the two elements in that position
with the appropriate probability.  However, our process selects arbitrary
pairs $(i,j)$ consisting of a single one and a single zero. Since we only
swap $(i,j)$ if they are neighboring, this may slow down the chain by
a factor of at most $k$.  

Since each exclusion process operates independently, the overall mixing
time will be roughly $n$ times the mixing time of each piece, slowed down by
the inverse probability of selecting that process.  Next, we will use
Theorems~\ref{crossproductthm} and~\ref{exclusion} to prove that $\mtree$ is
rapidly mixing.
\begin{theorem}
  \label{mtreemixing}
  If ${\bf P}$ has league structure $T$, then the mixing time of
    $\mtree$ under ${\bf P}$ satisfies 
  $$\tau_{tree}(\epsilon)=O(n^5\log (n/\epsilon)).$$
    If ${\bf P}$ is such that each $q_i>1/2$ is a constant, then
    $\tau_{tree}(\epsilon)=O(n^3\log n \log (n/\epsilon)).$
\end{theorem}
\begin{proof}
  In order to apply Theorem~\ref{crossproductthm} to the Markov chain
  $\mtree$, we note that for a node with $k_1$ ones and $k_2$ zeroes
  ($k=k_1+k_2$), the probability of selecting that node is
  ${k_1k_2}/{\binom{n}{2}}$. Since $M=n-1$, Theorem~\ref{crossproductthm}
  implies
  $$\tau(\epsilon)\leq \frac{n(n-1)}{k_1k_2} k^4
  \ln(2nk_1k_2/\epsilon)=O(n^5\log(n/\epsilon)).$$
  Of course, if all of the chains have probabilities that are bounded away
  from $1/2$, then we can use the second bound from Theorem~\ref{exclusion} to
  obtain 
  \begin{align*}
    \tau(\epsilon)&\leq  \frac{n(n-1)}{k_1k_2} k^2(\min\{k_1,k_2\} + \log
      k)\log(2n/\epsilon)\\
    &\leq  \frac{n(n-1)k^2}{\max\{k_1,k_2\}}
      \left(1+ \frac{\log k}{\min\{k_1,k_2\}} \right)\log(2n/\epsilon).
  \end{align*}
  There are two cases to consider.  Let $0<c<1$.  If
  $\min\{k_1,k_2\}\geq c\log k$ then
  $$\tau(\epsilon)\leq\frac{n(n-1)k^2}{k/2} (1+ c)\log(2n/\epsilon))=
  O(n^3 \log(n/\epsilon)).$$
  Otherwise, $\max\{k_1,k_2\}>k-c\log k$, so since $k\leq n$,
  $$\tau(\epsilon)\leq \frac{n(n-1)k^2}{k-c\log k}
  (1+ \log k)\log(2n/\epsilon))=\frac{n(n-1)k}{1-\frac{c\log k}{k}}
  (1+ \log k)\log(2n/\epsilon))=O( n^3\log n \log(n/\epsilon)).$$
\end{proof}

\section{Bounding the mixing time of $\mnn$ for both classes}
\label{comparisonSec}

Our goal now is to use the comparison method to obtain bounds on the mixing
time of $\mnn$ in the settings of Sections~\ref{inversionSec}
and~\ref{treeSec} from the bounds on the mixing times of $\minv$ and $\mtree$.
When comparing the mixing times of $\mtree$ and $\mnn$, for example, the goal
is to show that a move $e=(\sigma,\beta)$ of $\mtree$, which is allowed to
transpose $i$ and $j$ that are not necessarily nearest neighbors, can be
simulated with a sequence of moves of $\mnn$.  Moreover, we must ensure that
our path does not go through transitions that are much smaller in weight than
$\min\{\pi(\sigma),\pi(\beta)\}$.  This type of argument is straightforward
for the moves of $\minv$, and gives some intuition for the more involved
argument to compare $\mtree$ with $\mnn$, which will follow in
Section~\ref{mtreetomnn}.  

In the next two sections, we assume that each $p_{i,j}$ is a constant less
than 1; this is to ensure a good comparison between the spectral gap and the
mixing time.  If this condition is not satisfied, then the proofs still go
through and will give a bound on the spectral gap, but will not provide a good
bound on the mixing time.

\subsection{Comparing $\minv$ with $\mnn$}
\label{minvtomnn}
First, we consider the setting of Section~\ref{inversionSec}, where $p_{i,j}$
depends on $\min\{i,j\}$.  
\begin{theorem}
  \label{BiasedInversionsMnn}
  Let $1/2\leq r_1, r_2,\ldots, r_{n-1}< 1$ be constants.  Assume ${\bf P}$ is
  defined by $p_{i,j}=r_i$ for $i<j$.  Then the mixing time of $\mnn$ on
  biased permutations under ${\bf P}$ is $O(n^8\log(n/\epsilon))$.
\end{theorem}
\noindent Here we are using the bound from Theorem~\ref{BiasedInversions}
part 2, and if each $p_{i,j}$ is bounded away from $1/2$ then we would get a
better bound of $O(n^7\log(n/\epsilon))$ using Theorem~\ref{BiasedInversions}
part 1.  Recall that for any $a,b\in [n]$, we defined
$\lambda_{a,b}=p_{a,b}/p_{b,a}$.

\begin{proof}
  In order to apply Theorem~\ref{Comparison}, we need to define, for any
  transition $e=(\sigma,\beta)$ of the Markov chain $\minv$, a sequence of
  transitions of $\mnn$.  Let $e $ be a transition of $\minv$ which performs a
  transposition on elements $\sigma(i)$ and $\sigma(j)$, where $i<j$.  Recall
  $\minv$ can only swap $\sigma(i)$ and $\sigma(j)$ if all the elements
  between them are smaller than both $\sigma(i)$ and $\sigma(j)$.  To obtain a
  sufficient bound on the congestion along each edge, we ensure that in each
  step of the path, we do not decrease the weight of the configuration.  This
  is easy to do; in the first stage, move $\sigma(i)$ to the right, one step
  at a time, until it swaps with $\sigma(j)$.  This removes an inversion of
  the type $(\sigma(i),\sigma(k))$ for every $i<k<j$, so clearly we have not
  decreased the weight of the configuration at any step.  Next, move
  $\sigma(j)$ to the left, one step at a time, until it reaches position $i$.
  This completes the move $e$, and at each step, we are adding back an
  inversion of the type $(\sigma(j),\sigma(k))$ for some $i<k<j$.  Since
  $\sigma(k)=\min\{\sigma(j),\sigma(k)\}= \min\{\sigma(i),\sigma(k)\}$, we
  have $p_{\sigma(k),\sigma(i)}=p_{\sigma(k),\sigma(j)}$ for every $i<k<j$, so
  in this stage we restore all the inversions destroyed in the first stage,
  for a net change of $\lambda_{\sigma(j),\sigma(i)}$.

  Given a transition $(\upsilon,\omega)$ of $\mnn$ we must upper bound the
  number of canonical paths $\gamma_{\sigma\beta}$ that use this edge, which
  we do by bounding the amount of information needed in addition to
  $(\upsilon,\omega)$ to determine $\sigma$ and $\beta$ uniquely.
  For moves in the first stage, all we need to remember is $\sigma(j)$,
  because we know $\sigma(i)$ (it is the element moving forward).  We also
  need to remember $i$ (that is, the original location of $\sigma(i)$).  Given
  this information along with $\upsilon$ and $\omega$ we can uniquely recover
  $(\sigma,\beta)$.  Thus there are at most $n^2$ paths which use any edge
  $(\upsilon,\omega)$.  Also, notice that the maximum length of any path is
  $2n$. 

  Next we bound the quantity $A$ which is needed to apply
  Theorem~\ref{Comparison}.  Recall that we have guaranteed that
  $\pi(\sigma)\leq\max\{ \pi(\upsilon),\pi(\omega)\}.$  Assume first that
  $\pi(\sigma)\leq \pi(\upsilon)$.  Then
  \begin{align*}
    A &= \max_{(\upsilon,\omega) \in  E(P)} \left
    \{\frac{1}{\pi(\upsilon)P(\upsilon,\omega)}\sum_{\Gamma(\upsilon,\omega)}|\gamma_{\sigma\beta}|\pi(\sigma)P'(\sigma,\beta)
    \right \}\\
    &\leq \max_{(\upsilon,\omega) \in E(P)}
    \sum_{\Gamma(\upsilon,\omega)}2n\frac{P'(\sigma,\beta)}{{P(\upsilon,\omega)}}\ 
    \leq \ \max_{(\upsilon,\omega) \in E(P)}
    \sum_{\Gamma(\upsilon,\omega)}2n\frac{1/(2n)}{\frac{\lambda}{(1+\lambda)(n-1)}}
    \ = \ O(n^3).
  \end{align*}
  If, on the other hand, $\pi(\sigma)\leq \pi(\omega)$, then we use detailed
  balance to obtain:
  \begin{align*}
    A &= \max_{(\upsilon,\omega) \in  E(P)} \left
    \{\frac{1}{\pi(\upsilon)P(\upsilon,\omega)}\sum_{\Gamma(\upsilon,\omega)}|\gamma_{\sigma\beta}|\pi(\sigma)P'(\sigma,\beta)
    \right \}\\
    &= \max_{(\upsilon,\omega) \in  E(P)} \left
    \{\frac{1}{\pi(\omega)P(\omega,\upsilon)}\sum_{\Gamma(\upsilon,\omega)}|\gamma_{\sigma\beta}|\pi(\sigma)P'(\sigma,\beta)
  \right \}\\
    &\leq \max_{(\upsilon,\omega) \in E(P)}
    \sum_{\Gamma(\upsilon,\omega)}2n\frac{P'(\sigma,\beta)}{{P(\omega,\upsilon)}}\\
    &\leq \max_{(\upsilon,\omega) \in E(P)}
    \sum_{\Gamma(\upsilon,\omega)}2n\frac{1/(2n)}{\frac{\lambda}{(1+\lambda)(n-1)}}
    \ = \ O(n^3).
  \end{align*}
  In either case, we have $A=O(n^3)$.  Let $\lambda =\min_{i<j} \lambda_{j,i}.$
  Then $\pi_* = \min_{\rho \in\Omega}\pi(\rho) \geq{
  \lambda^{\binom{n}{2}}}/{n!}$, so $\log(1/(\epsilon
  \pi_*))=O(n^2\log\epsilon^{-1})$, since each $p_{i,j}$ bounded away from 1
  implies $\lambda$ is a positive constant.  Applying Theorem~\ref{Comparison}
  proves that $\tau_{nn}(\epsilon) = O(n^8\log(n/\epsilon))$. 
\end{proof}

\subsection{Comparing $\mtree$ with $\mnn$}\label{mtreetomnn}
In this section we show that $\mnn$ is rapidly mixing when ${\bf P}$ has
league structure and is \emph{weakly monotone}:

\begin{definition}
  \label{weakmonotone}
  The set ${\bf P}$ is \emph{weakly monotone} if 
  properties 1 and \emph{either} 2 or 3 are satisfied.
  \begin{enumerate}
    \item $p_{i,j} \geq 1/2$ for all $1 \leq i < j \leq n$, and
    \item $p_{i,j+1} \geq p_{i,j}$  for all $1 \leq i < j \leq n-1$ or
    \item $p_{i-1,j} \geq p_{i,j}$  for all $2 \leq i < j \leq n$.
  \end{enumerate}
\end{definition}
\noindent We note that if ${\bf P}$ satisfies all three properties then it is
\emph{monotone}, as defined by Fill~\cite{F03a}.

The comparison proof in this setting is similar to that of
Section~\ref{minvtomnn}, except that there may be elements between $\sigma(i)$
and $\sigma(j)$ that are larger than both and elements that are smaller than
both.  This poses a problem, because we may not be able to move $\sigma(i)$
past all the elements between them without greatly decreasing the weight.
However, when ${\bf P}$ is weakly monotone, we can introduce a trick to get
around this problem. At a high level, we shift the elements between
$\sigma(i)$ and $\sigma(j)$ that are smaller than $\sigma(i)$ and $\sigma(j)$
to the left in a special way, increasing the weight of the configuration in
such a way that when we move $\sigma(i)$ to the right, the weight never goes
below $\min\{\pi(\sigma),\pi(\beta)\}$.  Specifically, we prove the following
theorem.

\begin{theorem}
  \label{treecomp}
  If ${\bf P}$ has league structure, is weakly monotone and is such that
  $p_{i,j}$ is a constant less than 1 for all $i,j$, then the mixing time of
  $\mnn$ satisfies $\tau_{nn}(\epsilon)= O(n^9\log(n/\epsilon)).$
\end{theorem} 
\noindent Again, we are assuming the worst case bound on the mixing time of
$\mtree$ given in Theorem~\ref{mtreemixing}, and if each $p_{i,j}$ is bounded
away from $1/2$ then we would get a better bound.

\begin{proof}
  Throughout this proof we assume that ${\bf P}$ satisfies properties $1$ and
  $2$ of the weakly monotone definition.  If instead ${\bf P}$ satisfies
  property $3$, then the proof is very similar. 
  In order to apply Theorem~\ref{Comparison} to relate the mixing time of
  $\mnn$ to the mixing time of $\mtree$ we need to define for each transition
  of $\mtree$ a canonical path using transitions of $\mnn$.  Let $e =
  (\sigma,\beta)$ be a transition of $\mtree$ which performs a transposition
  of elements $\sigma(i)$ and $\sigma(j)$ where $i<j$.  If there are no
  elements between $\sigma(i)$ and $\sigma(j)$ then $e$ is already a
  transition of $\mnn$ and we are done. Otherwise, $\sigma$ contains the
  string $\sigma(i), \sigma(i+1), ... \sigma(j-1), \sigma(j)$ and $\beta$
  contains $\sigma(j), \sigma(i+1), ... \sigma(j-1), \sigma(i)$.  From the
  definition of $\mtree$ we know that for each $\sigma(k)$, $k \in [i+1,j-1]$,
  either $\sigma(k) > \sigma(i),\sigma(j)$ or $\sigma(k) <
  \sigma(i),\sigma(j)$.  Define $S=\{\sigma(k): \sigma_k < \sigma(i),
  \sigma(j)\}$ and $B=\{\sigma(k): \sigma_k > \sigma(i),\sigma(j)\}$.  To
  obtain a good bound on the congestion along each edge we must ensure that
  the weight of the configurations on the path are not smaller than the weight
  of $\sigma$.
  To this end, we define three stages in our path from $\sigma$ to $\beta$.
  In the first, we shift the elements of $S$ to the left, removing an
  inversion with each element of $B$.  In the second stage we move $\sigma(i)$
  next to $\sigma(j)$ and in the third stage we move $\sigma(j)$ to
  $\sigma(i)$'s original location.  Finally, we shift the elements of $S$ to
  the right to return them to their original locations.  See
  Figure~\ref{Ahighlevelpath}. 
   
  \begin{figure}
    \center{$
      \begin{array}{ccccccccc}
      \textbf{5} & 8 & 9 & \underline{2} & 10 & \underline{3} & \underline{4} &
        \underline{1} & \textbf{7} \\
      \textbf{5} & \underline{2} & 8 & 9 & \underline{3} & 10 & \underline{4} &
        \underline{1} & \textbf{7} \\
      \underline{2} & 8 & 9 & \underline{3} & 10 & \underline{4} &
        \underline{1} & \textbf{5} & \textbf{7} \\
      \underline{2} & 8 & 9 & \underline{3} & 10 & \underline{4} & \underline{1}
        & \textbf{7} & \textbf{5} \\
      \textbf{7} & \underline{2} & 8 & 9 & \underline{3} & 10 & \underline{4} &
        \underline{1} & \textbf{5} \\
      \textbf{7} & 8 & 9 & \underline{2} & 10 & \underline{3} & \underline{4} &
        \underline{1} & \textbf{5} \\
      \end{array}$
    }
    \caption{The canonical path for transposing 5 and 7.  Notice
      that the elements in $S$ are underlined.}
          \label{Ahighlevelpath}
  \end{figure}
   
  \vspace{.05in}
  \noindent  \textbf{Stage 1:} In this stage, for each $b\in B$, we remove an
  inversion involving $b$ by shifting an element of $S$ to the left past $b$.
  More precisely, if $\sigma(j-1)\in B$, shift $\sigma(j)$ to the left until
  an element from $S$ is immediately to the left of $\sigma(j)$.  Next,
  starting at the right-most element in $S$ and moving left, for each
  $\sigma(k) \in S$ such that $\sigma(k-1) \in B$, move $\sigma(k)$ to the
  left one swap at a time until $\sigma(k)$ has an element from $S$ or
  $\sigma(i)$ on its immediate left (see Figure \ref{Astage1}).  Notice that
  for each element $b \in B$ we have removed exactly one $(b,\sigma(k))$
  inversion where $\sigma(k) \in S\cup \sigma(j)$.  

  \vspace{.05in}
  \noindent \textbf{Stage 2:} Next perform a series of nearest neighbor swaps
  to move $\sigma(i)$ to the right until it is in position $j$ (the original
  position occupied by $\sigma(j)$ in $\sigma$, see Figure \ref{Astage2}).
  While we have created a $(b,\sigma(i))$ inversion for each element $b \in
  B$, we claim that the weight has not decreased from the original weight by
  more than a factor of $\lambda_{\sigma(j),\sigma(i)}$.  This is because in
  Stage 1, for each element $b\in B$, we removed a $(b,s)$ inversion for some
  $s\in S\cup \sigma(j)$.  Assume first that $s\in S$.  Then since
  $b>\sigma(i)>s$, it follows that $p_{b, \sigma(i)} \geq p_{b, s}$ for all
  $s\in S$ since the ${\bf P}$ are weakly monotone; thus, for each $b$ we
  introduce a multiplicative factor of $\lambda_{b,\sigma(i)}/\lambda_{b,s}
  \geq 1$.  On the other hand, if $s=\sigma(j)$ then recall $p_{b,\sigma(j)} =
  p_{b,\sigma(i)}$ because $b$ is not a descendant of $\sigma(i)\vee
  \sigma(j)$ in the tree $T$.  Hence the current configuration has weight at
  least $\lambda_{\sigma(j),\sigma(i)}\pi(\sigma)$.  Since
  $\lambda_{\sigma(j),\sigma(i)}$ is also the ratio of $\pi(\sigma)$ and
  $\pi(\beta)$, it follows that the weight at every step of Stage 2 does not
  go below $\min\{\pi(\sigma),\pi(\beta)\}$.  For each $\sigma(k) \in S$ we
  have also removed a $(\sigma(k),\sigma(j))$ inversion, which can only
  increase the weight of the configuration.

\begin{figure}[!htb]
  \centering
  \subfloat[]{
  \begin{minipage}[c][.35\width]{0.5\textwidth}
  \centering%	
  $
    \begin{array}{ccccccccc}
      \textbf{5} & 8 & 9 & \underline{2} & 10 & \underline{3} & \underline{4}
        & \underline{1} & \textbf{7} \\
      \textbf{5} & 8 & 9 & \underline{2} & \underline{3} & 10 & \underline{4} &
        \underline{1} & \textbf{7} \\
      \textbf{5} & 8 & \underline{2} & 9  & \underline{3} & 10 & \underline{4} &
        \underline{1} & \textbf{7} \\
      \textbf{5} & \underline{2} & 8 & 9 & \underline{3} & 10 & \underline{4} &
        \underline{1} & \textbf{7} \\
    \end{array}$
\end{minipage}
\label{Astage1}
}
\subfloat[]{
\begin{minipage}[c][.35\width]{%
0.5\textwidth} \centering%	
$  \begin{array}{ccccccccc} 
        \textbf{5} & \underline{2} & 8 & 9 & \underline{3} & 10 & \underline{4}
          & \underline{1} & \textbf{7} \\
        \underline{2} & \textbf{5} & 8 & 9 & \underline{3} & 10 & \underline{4}
          & \underline{1} & \textbf{7} \\
        \underline{2} & 8 & \textbf{5} & 9 & \underline{3} & 10 & \underline{4}
          & \underline{1} & \textbf{7} \\
        & & & & \vdots \\
        \underline{2} & 8 & 9 & \underline{3} & 10 & \underline{4} &
          \underline{1} & \textbf{5} & \textbf{7} \\
        \underline{2} & 8 & 9 & \underline{3} & 10 & \underline{4} &
          \underline{1} & \textbf{7} & \textbf{5} \\
      \end{array}$
\end{minipage}
\label{Astage2}
}
\caption{Stages 1 and 2 of the canonical path for transposing 5 and 7.}
\end{figure}

\vspace{.05in}
  \noindent  \textbf{Stage 3:} Perform a series of nearest neighbor swaps to
  move $\sigma(j)$ to the left until it is in the same position $\sigma(i)$
  was originally.  While we created an $(\sigma(k),\sigma(j))$ inversion for
  each $\sigma(k) \in S$, these inversions have the same weight as the
  $(\sigma(i),\sigma(k))$ inversion we removed in Stage 2.  In addition we
  have removed an $(\sigma(l),\sigma(j))$ inversion for each $\sigma(l) \in
  B$.  

\vspace{.05in}
 \noindent \textbf{Stage 4:} Finally we want to return the elements in $S$ and
 $B$ to their original position.  Starting with the left-most element in $S$
 that was moved in Stage 1, perform the nearest neighbor swaps to the right
 necessary to return it to its original position.  
\vspace{.05in}

It's clear from the definition of the stages that along any path the weight of
a configuration never decreases below the weight of
$\min(\pi(\sigma),\pi(\beta))$.  Given a transition $(\upsilon,\omega)$ of
$\mnn$ we must upper bound the number of canonical paths
$\gamma_{\sigma\beta}$ that use this edge.  Thus, we analyze the amount of
information needed in addition to $(z,w)$ to determine $\sigma$ and $\beta$
uniquely.  First we record whether $(\sigma,\beta)$ is already a nearest
neighbor transition or which stage we are in.  Next for any of the 4 stages we
record the original location of $\sigma(i)$ and $\sigma(j)$.  Given this
information, along with $\upsilon$ and $\omega$, we can uniquely recover
$(\sigma,\beta)$.  Hence, there are at most $4n^2$ paths through any edge
$(\upsilon,\omega)$.  Also, note that the maximum length of any path is $4n$.

Next we bound the quantity $A$ which is needed to apply
Theorem~\ref{Comparison}.  Recall that for each transition $(\upsilon,
\omega)$ of the path $\gamma_{\sigma,\beta}$, we have guaranteed that
$\pi(\upsilon)\geq\min\{ \pi(\sigma),\pi(\beta)\}.$  Assume first that
$\pi(\upsilon)\geq \pi(\sigma)$.  Then
\begin{align*}
  A &= \max_{(\upsilon,\omega) \in  E(P)} \left
  \{\frac{1}{\pi(\upsilon)P(\upsilon,\omega)}\sum_{\Gamma(\upsilon,\omega)}|\gamma_{\sigma\beta}|\pi(\sigma)P'(\sigma,\beta)
  \right \}\\
  &\leq \max_{(\upsilon,\omega) \in E(P)}
  \sum_{\Gamma(\upsilon,\omega)}2n\frac{P'(\sigma,\beta)}{{P(\upsilon,\omega)}}\
  \leq \max_{(\upsilon,\omega) \in E(P)}
  \sum_{\Gamma(\upsilon,\omega)}2n\frac{1/\binom{n}{2}}{\frac{\lambda}{(1+\lambda)(n-1)}}
  \ = \ O(n^2).
\end{align*}
If, on the other hand, $\pi(\upsilon)\geq \pi(\beta)$, then we use detailed
balance to obtain:
\begin{align*}
  A &= \max_{(\upsilon,\omega) \in  E(P)} \left
  \{\frac{1}{\pi(\upsilon)P(\upsilon,\omega)}\sum_{\Gamma(\upsilon,\omega)}|\gamma_{\sigma\beta}|\pi(\sigma)P'(\sigma,\beta)
  \right \}\\
  &= \max_{(\upsilon,\omega) \in  E(P)} \left
  \{\frac{1}{\pi(\upsilon)P(\upsilon,\omega)}\sum_{\Gamma(\upsilon,\omega)}|\gamma_{\sigma\beta}|\pi(\beta)P'(\beta,\sigma)
  \right \}\\
  &\leq \max_{(\upsilon,\omega) \in E(P)}
  \sum_{\Gamma(\upsilon,\omega)}2n\frac{P'(\beta,\sigma)}{{P(\upsilon,\omega)}}\\
  &\leq \max_{(\upsilon,\omega) \in E(P)}
  \sum_{\Gamma(\upsilon,\omega)}2n\frac{1/\binom{n}{2}}{\frac{\lambda}{(1+\lambda)(n-1)}}
  \ = \ O(n^2).
\end{align*}
In either case, we have $A=O(n^2)$.  Let $\lambda =\min_{i<j} \lambda_{j,i}.$
Then $\pi_* = \min_{\rho \in\Omega}\pi(\rho) \geq{
\lambda^{\binom{n}{2}}}/{n!}$, so $\log(1/(\epsilon
\pi_*))=O(n^2\log\epsilon^{-1})$, as above.  Applying Theorem~\ref{Comparison}
proves that $\tau_{nn}(\epsilon) = O(n^9\log(n/\epsilon))$. 
\end{proof}

\begin{Remark}
  By repeating Stage 1 of the path a constant number of times, it is possible
  to relax the weakly monotone condition slightly if we are satisfied with a
  polynomial bound on the mixing time.
\end{Remark}

\section{Slow Mixing of $\mnn$}
\label{slow}
We conclude by showing that while $\mnn$ is rapidly mixing for two large,
interesting classes of inputs, this is not true in general. In particular, we
show that there are positively biased permutations for which the chain $\mnn$
requires exponential time to converge to equilibrium.  This disproves the
conjecture that the chain will always be fast when ${\bf P}$ satisfies $p_{ij}
\geq 1/2$ for all $i < j$.  

Our example comes from sampling staircase walks with fluctuating bias, which
were examined in~\cite{GPR09} and~\cite{rs}.  Staircase walks are sequences of
$n$ ones and $n$ zeros, which correspond to paths from $(0,n)$ to $(n,0)$,
where each 1 represents a step to the right and each 0 represents a step down
(see Figure~\ref{staircase}b). For ease of notation in the following proof, we
replace the zeroes by negative ones.  In~\cite{rs}, Randall and Streib
examined the Markov chain which attempts to swap a neighboring $(1,-1)$ pair,
which essentially adds or removes a unit square from the region below the
walk, with probability depending on the position of that unit square.  We will
show that for our choice of ${\bf P}$, permutations are equivalent to
staircase walks, and hence the proof that the Markov chain on staircase walks
is slow applies in our setting as well.

Suppose, for ease of notation, that we are sampling permutations with $2n$
entries (having an odd number of elements will not cause qualitatively
different behavior).  Let $M=n-\sqrt{n}$, $\epsilon=1/(16n+2)$, and
$\frac{1}{65}<\delta<\frac{1}{2}$ be a constant to be defined later.  For
$i<j\leq n$ or $n<i<j$, $p_{i,j}=1$, ensuring that once the elements
$1,2,\ldots, n$ get in order, they stay in order (and similarly for the
elements $n+1,n+2,\ldots, 2n$).  The $p_{i,j}$ values for $i\leq n<j$ are
defined as follows (see Figure~\ref{staircase}a): $$p_{i,j}=
\begin{cases}
  1-\delta  & \text{ if $i-j+2n+1\geq n+M$};\\ 
  \frac{1}{2}+\epsilon & \text{otherwise}.
\end{cases}$$
\noindent Since the smallest (largest) $n$ elements of the biased permutation
never change order once they get put into increasing order, permutations with
these elements out of order have zero stationary probability.  Hence we can
represent the smallest $n$ numbers as ones and the largest $n$ numbers as
negative ones, assuming that within each class the elements are in increasing
order.  Given a permutation $\sigma$, let $f(\sigma)$ be the sequence of 1's
and -1's such that $f(\sigma)_i =1$ if $i\leq n$ and $-1$ otherwise.  Then if
$\sigma$ is such that elements $1,2,\ldots, n$ and elements $n+1,n+2,\ldots,
2n$ are each in order, $f(\sigma)$ maps $\sigma$ uniquely to a staircase walk.
For example, the permutation $\sigma = (5,1,7,8,4,3,6,2)$ maps to
$f(\sigma)=(-1,1,-1,-1,1,1,-1,1)$.
The probability that an adjacent 1 and -1 swap in $\mnn$ then depends on how
many 1's and -1's occur before that point in the permutation.  Specifically,
if element $i$ is $-1$ and element $i+1$ is 1 then we swap them with
probability $\frac{1}{2}+\epsilon$ if the number of 1's occurring before
position $i$ plus the number of -1's occurring after $i+1$ is less than
$n+M-1$.  Otherwise, they swap with probability $1-\delta$.  Equivalently, the
probability of adding a unit square at position $v=(x,y)$, which is called the
\emph{bias at $v=(x,y)$}, is $\frac{1}{2}+\epsilon$ if $x+y\leq n+M$, and
$1-\delta$ otherwise; see Figure~\ref{staircase}a.  We will show that in this
case, the Markov chain is slow.  The idea is that in the stationary
distribution, there is a good chance that the positive and negative ones will
be well-mixed, since this is a high entropy situation.  However, the identity
permutation also has high weight, and the parameters are chosen so that the
entropy of the well-mixed permutations balances with the energy of the maximum
(identity) permutation, and that to get between them is not very likely (low
entropy and low energy).

We identify sets $S_1,S_2,S_3$ such that $\pi(S_2)$ is exponentially smaller
than both $\pi(S_1)$ and $\pi(S_3)$, but to get between $S_1$ and $S_3$,
$\mnn$ and $\mt$ must pass through $S_2$, the cut.  Then we use the
\emph{conductance} to prove $\mnn$ and $\mt$ are slowly mixing.  For an
ergodic Markov chain with stationary distribution $\pi$, the conductance is
$$\Phi = \min_{{S\subseteq\Omega}\atop{ \pi(S) \leq 1/2}}\sum_{s_1\in S,
s_2\in \bar{S}}\pi(s_1)\prob(s_1,s_2)/{\pi(S)},$$ and we will show that the
bad cut $(S_1,S_2,S_3)$ defined in Section~\ref{slow} implies that $\Phi$ is
exponentially small.  The following theorem relates the conductance and mixing
time (see \cite{JS89}).
\begin{theorem}
  \label{conductance}
  For any Markov chain with conductance $\Phi$, \
  $\tau \geq ({4\Phi})^{-1} - 1/2.$
\end{theorem}
We are now ready to prove the main theorem from this section.

\begin{theorem}\label{slowmixing} 
  There exists a set ${\bf P}$ 
  for which $\mnn$ has mixing time $\tau(\epsilon) \geq
  e^{n/24}/4 - 1/2.$
\end{theorem}

\begin{figure}[htb] \centering
 \includegraphics[scale=.08]{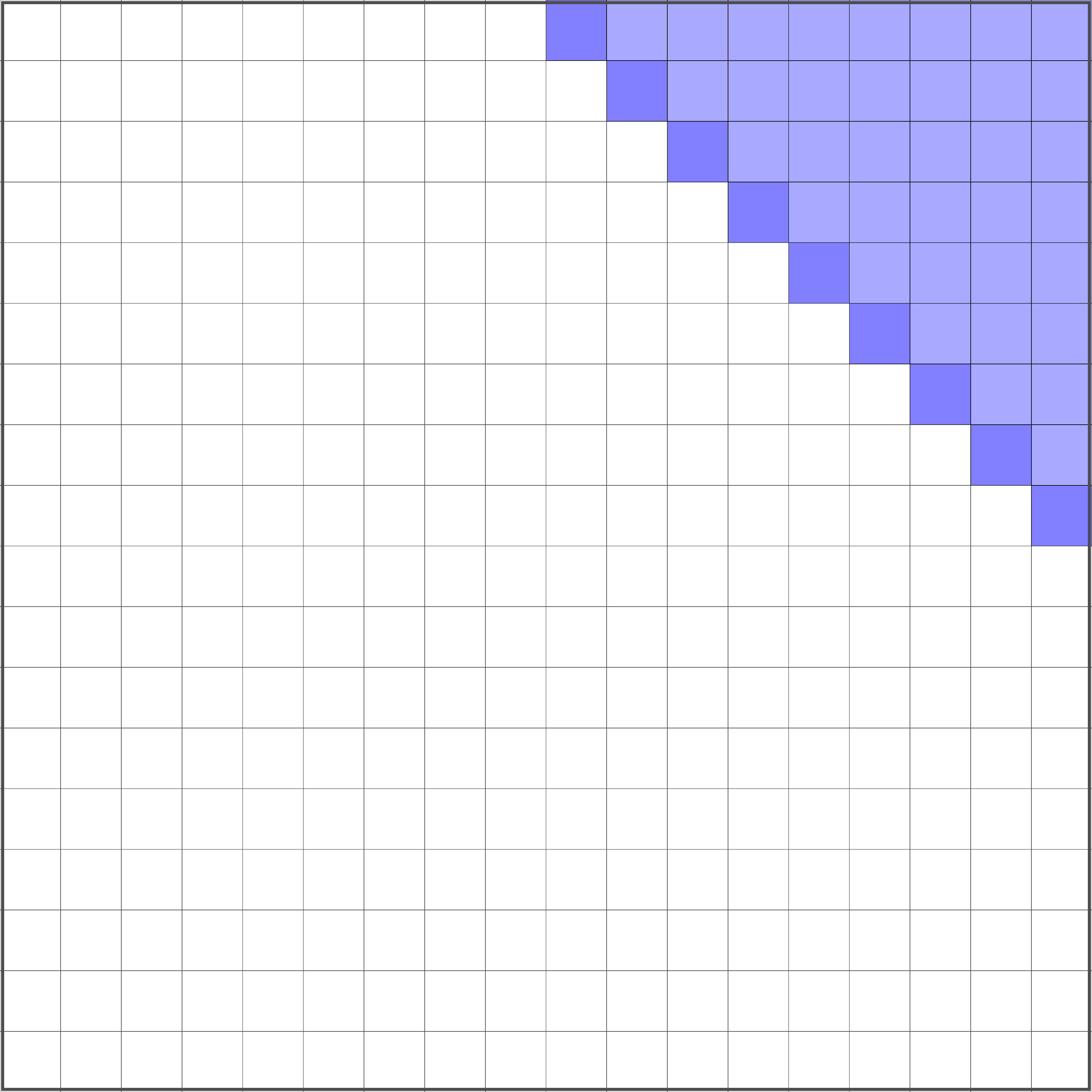}
  \put(-80,50){$\frac{1}{2} + \epsilon$} 
  \put(-30,95){$1-\delta$} \tiny
  \put(-90,118){$M$} 
  \put(-40,118){$n-M$} 
  \hspace{1in}
   \includegraphics[scale=.08]{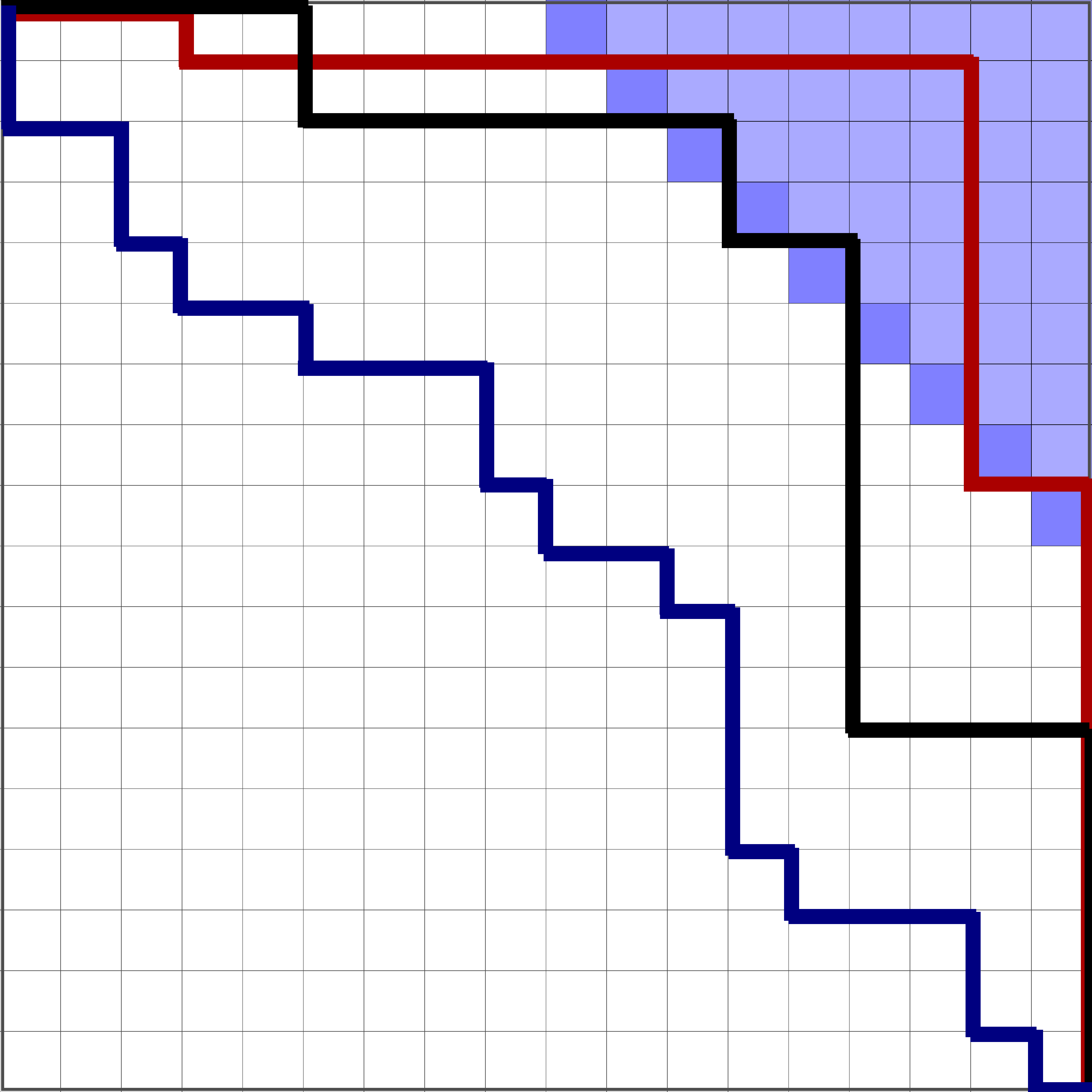}
  \caption{(a) Fluctuating bias with exponential mixing time. (b) Staircase
  walks in $S_1, S_2,$ and $S_3$.}
  \label{staircase}
\end{figure}

\begin{proof}
  For a staircase walk $\sigma$ consisting of a sequence of steps
  $\sigma_{i}\in\{\pm 1\}$, define the \emph{height of $\sigma_i$} as
  $\sum_{j\leq i} \sigma_{j}$, and let $\max(\sigma) $ be the maximum height
  of $\sigma_i$ over all $1\leq i\leq 2n$.  Let $S_1$ be the set of
  configurations $\sigma$ such that $\max(\sigma)< n+M$, $S_2$ the set of
  configurations such that $\max(\sigma)= n+M$, and $S_3$ the set of
  configurations such that $\max(\sigma)> n+M$.  That is, $S_1$ is the set of
  configurations that never reach the dark blue diagonal in
  Figure~\ref{staircase}b, $S_2$ is the set whose maximum peak is on the dark
  blue line, and $S_3$ is the set which crosses that line and contains squares
  in the light blue triangle.  Clearly to move from $S_1$ to $S_3$, the Markov
  chain must go through $S_2$.

  Define $\gamma=(1/2+\epsilon)/(1/2-\epsilon)$, which is the ratio of two
  configurations that differ by swapping a $(1,-1)$ pair with probability
  $\frac{1}{2}+\epsilon$.  By the definition of $\epsilon$, we have
  $\gamma=1+\frac{1}{4n}$.  Let $\xi = (1-\delta)/\delta$, which is the ratio
  of two configurations that differ by swapping a $(1,-1)$ pair with
  probability $1-\delta$.  Finally, let $b(\sigma)$ be the number of tiles
  below the diagonal $M$ in $\sigma$ and $a(\sigma)$ be the number of tiles
  above the diagonal $M$ in $\sigma$.  Then by detailed balance,
  $\pi(S_1)={Z}^{-1}\sum_{\sigma\in S_1} \gamma^{a(\sigma)},
  \pi(S_2)={Z}^{-1}\sum_{\sigma\in S_2}\gamma^{a(\sigma)},$ and $\pi(S_3)
  =Z^{-1} \sum_{\sigma\in S_3 } \gamma^{b(\sigma)} \xi^{a(\sigma)},$ where $Z$
  is a normalizing constant.  We will show that there exists a constant
  $\frac{1}{65}<\delta <\frac{1}{2}$ such that $\pi(S_2)$ is exponentially
  smaller than both $\pi(S_1)$ and $\pi(S_3)$, which will have equal weight.
  
  First we show that $\pi(S_2)$ is exponentially smaller than $\pi(S_1)$ for
  all values of $\delta$.  Since there are at most $n^2-(n-M)^2/2=n^2-n/2$
  tiles with weight $\gamma$ in any $\sigma\in S_2$, we have $$\pi(S_2)\ =\
  {Z}^{-1}\sum_{\sigma\in S_2}\gamma^{a(\sigma)}\ \leq\ {Z}^{-1}
  \gamma^{n^2-n/2}|S_2| \ \leq \ {Z}^{-1} e^{n/4-1/8}|S_2|, $$ since
  $\gamma^{n^2-n/2} = (1+  \frac{1}{4n})^{n^2-n/2}  \leq  e^{n/4 - 1/8}. $

  Next we will bound $|S_2\cup S_3|$, which in turn provides an upper bound on
  $|S_2|$.  The unbiased Markov chain is equivalent to a simple random walk
  $W_{2n}=X_1+X_2+\cdots +X_{2n}=0$, where $X_i\in \{+1,-1\}$ and where a $+1$
  represents a step to the right and a $-1$ represents a step down.  We call
  this random walk {\it tethered} since it is required to end at $0$ after
  $2n$ steps.  Compare walk $W_{2n}$ with the untethered simple random walk
  $W_{2n}'=X_1'+X'_2+\ldots + X'_{2n}$.
  \begin{eqnarray*}
    P\left(\max_{1\leq t\leq 2n} W_t\geq M\right)&=&P\left(\max_{1\leq t\leq
    2n} W'_t\geq M ~|~ W'_{2n}=0\right)\\
    &=&\frac{ P\left(\max_{1\leq t\leq 2n} W'_t\geq M\right)}{ P(W'_{2n}=0)
    }\\
    &=&  \frac{2^{2n}}{\binom{2n}{n}}P\left(\max_{1\leq t\leq 2n} W'_t\geq
    M\right)\\
    &\approx&\sqrt{\pi n}~~P\left(\max_{1\leq t\leq 2n} W'_t\geq M\right).
  \end{eqnarray*}
  Since the $\{X'_i\}$ are independent, we can use Chernoff bounds to see that

  $$P\left(\max_{1\leq t\leq 2n}W'_t\geq M\right) \ \leq \  2nP(W'_{2n}\geq M) \ \leq \  2n e^{\frac{-M^2}{2n}}.$$
  Notice that $M^2/(2n) = (n-\sqrt{n})^2/(2n)= (\sqrt{n}-1)^2/2\geq n/3 $ for
  $n\geq 4$.
  Together these show that $P\left(\max_{1\leq t\leq 2n} W_t\geq M\right)\leq
  \sqrt{\pi} n^{3/2}e^{-n/3}$.  In particular, 
  $$|S_2\cup S_3|\leq \binom{2n}{n}\sqrt{\pi}n^{3/2} e^{-n/3}.$$
  Therefore we have
  \begin{eqnarray*}
    \frac{\pi(S_2)}{\pi(S_1)}&\leq &\frac{\frac{1}{Z} e^{n/4-1/8}|S_2\cup
    S_3|}{\frac{1}{Z} \left(\binom{2n}{n} - |S_2\cup S_3|\right)}\\
    &\leq&  e^{n/4-1/8} \left(\frac{\binom{2n}{n}}{|S_2\cup S_3|} -
    1\right)^{-1}\\
    &\leq&  e^{n/4-1/8}\left(\frac{\binom{2n}{n}}{
    \binom{2n}{n}\sqrt{\pi}n^{3/2} e^{-n/3}} - 1\right)^{-1}\\
    &\leq& \frac{e^{n/4-1/8}\sqrt{\pi}n^{3/2}}{e^{n/3} - \sqrt{\pi}n^{3/2}} \
    < \ e^{-n/24},
  \end{eqnarray*}
  for large enough $n$.  Therefore, $\pi(S_2)$ is exponentially smaller than
  $\pi(S_1)$ for every value of $\delta$.

  Our goal is to show that there exists a value of $\delta$ for which
  $\pi(S_3)=\pi(S_1)$, which will imply that $\pi(S_2)$ is also exponentially
  smaller than $\pi(S_3)$, and hence the set $S_2$ forms a bad cut, regardless
  of which state the Markov chain begins in.  To find this value of $\delta$, we
  will rely on the continuity of the function $f(\xi) = Z\pi(S_3) - Z\pi(S_1)$
  with respect to $\xi=(1-\delta)/\delta$.  Notice that $Z\pi(S_1)$ is constant
  with respect to $\xi$ and $Z\pi(S_3) = \sum_{\sigma\in S_3 }
  \gamma^{b(\sigma)} \xi^{a(\sigma)}$ is just a polynomial in $\xi$.  Therefore
  $Z\pi(S_3)$ is continuous in $\xi$ and hence $f(\xi)$
  is also continuous with respect to $\xi$.  Moreover, when $\xi = \gamma$,
  clearly $Z\pi(S_3)< Z\pi(S_1)$, so $f(\gamma)<0$.  We will show that
  $f(4e^2)>0$, and so by continuity we will conclude that there exists a value
  of $\xi$ satisfying $\gamma< \xi < 4e^2$ for which $f(\xi)=0$ and
  $Z\pi(S_3)=Z\pi(S_1)$.  Clearly this implies that for this choice of $\xi$,
  $\pi(S_3)=\pi(S_1)$, as desired.  To obtain the corresponding value of
  $\delta$, we notice that $\delta = 1/(\xi+1)$.  In particular, $\delta$ is a
  constant satisfying $\frac{1}{65}<\delta<\frac{1}{2}.$

  Thus it remains to show that $f(4e^2)>0$.  First we notice that since the
  maximal tiling is in $S_3$, $\pi(S_3)\geq {Z}^{-1}
  \gamma^{n^2-\frac{(n-M)^2}{2}}\xi^{\frac{(n-M)^2}{2}}.$ Also,
  $$\pi(S_1)={Z}^{-1}\sum_{\sigma\in S_1} \gamma^{a(\sigma)} \ < \
  Z^{-1}\binom{2n}{n} \gamma^{n^2-\frac{(n-M)^2}{2}}.$$
  Therefore
  $$\pi(S_1)/\pi(S_3)< \frac{\binom{2n}{n}}{\xi^{\frac{(n-M)^2}{2}}}
  \leq {(2e)^n}{\xi^{-n/2}} = 1 $$ since $\xi=4e^2$.  Hence $f(4e^2) =
  Z\pi(S_3) - Z\pi(S_1)> Z\pi(S_3) - Z\pi(S_3)=0,$ as desired.

  Thus, the conductance satisfies 
  $$\Phi \leq
  \frac{1}{\pi(S_1)}  \sum_{x\in S_1}\pi(x)\sum_{y\in S_2} P(x,y)\leq  
  \frac{1}{\pi(S_1)} \sum_{x\in
  S_1}{\pi(x)}\pi(S_2)\leq 
  e^{-n/24}.$$ 
    
  \noindent Hence, by Theorem~\ref{conductance}, the mixing time satisfies
  $$\tau\geq ({4e^{-n/24}})^{-1} - 1/2 \geq e^{n/24}/4-1/2.$$
\end{proof}

In fact, this proof can be extended to the more general Markov chain where we
can swap any $1$ with any $0$, as long as we maintain the correct stationary
distribution.  This is easy to see, because any move that swaps a single 1
with a single 0 can only change the maximum height by at most 2 (see
Figure~\ref{longswaps}).  If we expand $S_2$ to include all configurations
with maximum height $n+M$ or $n+M+1$, $\pi(S_2)$ is still exponentially
smaller than $\pi(S_1)$ and $\pi(S_3)$.  Hence the Markov chain over
permutations that can make arbitrary transpositions can still take exponential
time to converge.

\begin{figure}[htb]
  \centering \includegraphics[scale=.08]{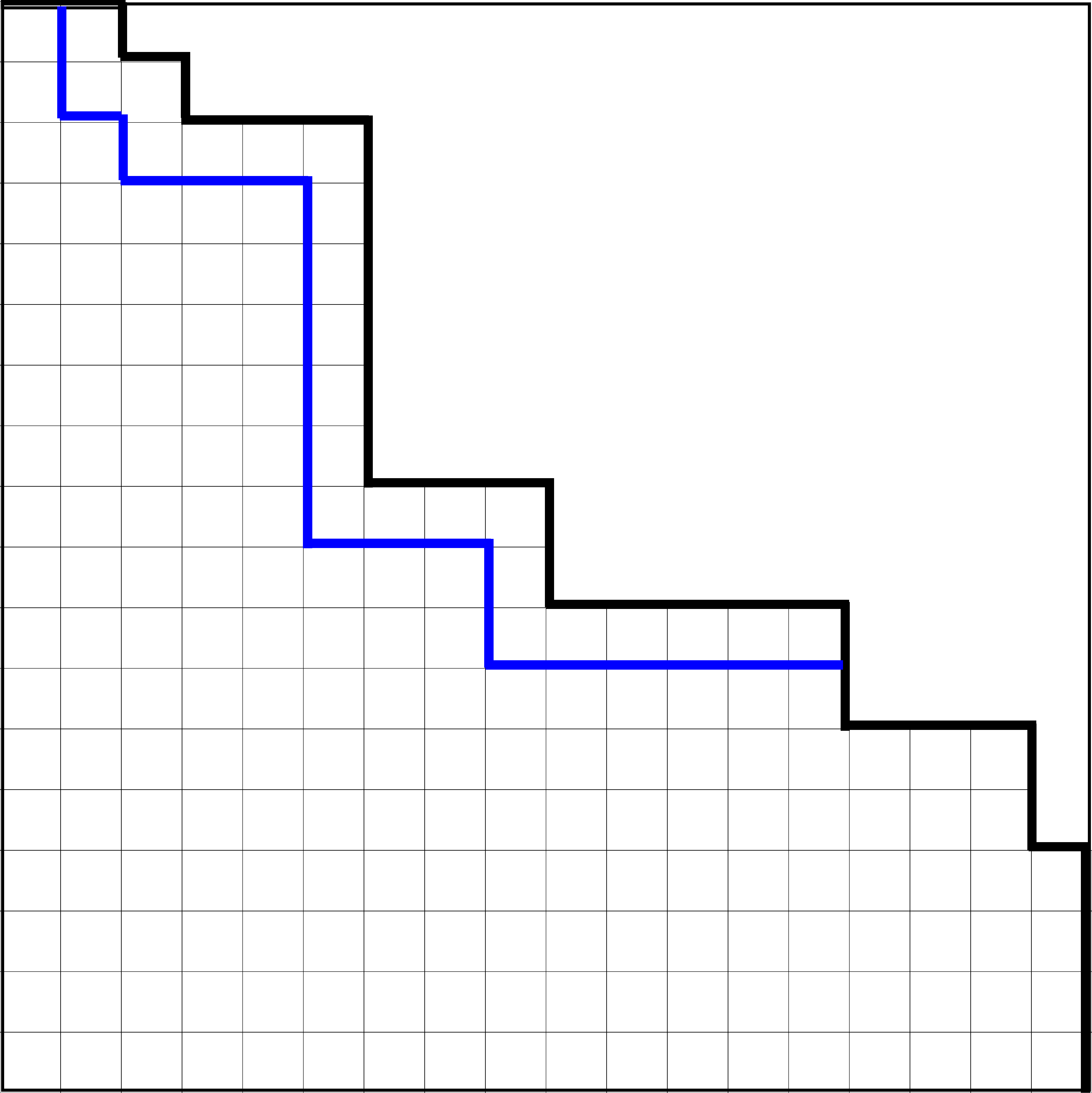}
  \caption{A move that swaps an arbitrary $(1,0)$ pair.}
  \label{longswaps}
\end{figure}

\section{Analyzing a Product of Markov chains} \label{productsection}
For each of our positive results, we showed that the Markov chain in question
can be decomposed into $M$ independent Markov chains.  Since each Markov chain
$\m_i$ operates independently, the overall mixing time will be roughly $M$
times the mixing time of each piece, slowed down by the inverse probability of
selecting that process. Similar results have been proved before (e.g., see
\cite{BBHM05, BR04}) in other settings.  We include the proof for
completeness.

\begin{theorem}
  \label{crossproductthm}
  Suppose the Markov chain $\m$ is a product of $M$ independent Markov
  chains $\m_1,\m_2,\ldots, \m_M$, where $\m$ updates $\m_i$ with probability
  $p_i$, where $\sum_i p_i=1$. If $\tau_i(\epsilon)$ is the mixing time for
  $\m_i$ and $\tau_i(\epsilon)\geq 4\ln\epsilon$ for each $i$, then
  \[\tau(\epsilon)\leq \max_{i=1,2,\ldots, M} \ \frac{2}{p_i}\ 
  \tau_i\left(\frac{\epsilon}{2M}\right).\]
\end{theorem}

\begin{proof}
  Suppose the Markov chain $\m$ has transition matrix $P$, and each $\m_i$ has
  transition matrix $P_i$ and state space $\Omega_i$.  Let
  $B_i=p_iP_i+(1-p_i)I$, where $I$ is the identity matrix of the same size as
  $P_i$, be the transition matrix of $\m_i$, slowed down by the probability
  $p_i$ of selecting $\m_i$.  First we show that the total variation distance
  satisfies $$1+2d_{tv}(P^t,\pi)\leq \prod_{i} (1+2d_{tv}(B_i^t,\pi_i)).$$ To
  show this, notice that for $x=(x_1,x_2,\ldots, x_M), y=(y_1,y_2,\ldots,
  y_M)\in \Omega$, $P^t(x,y)=\prod_i B_i^t(x_i,y_i)$.  Let $\epsilon_i
  (x_i,y_i)= B_i^t(x_i,y_i) - \pi_i(y_i)$ and for any $x_i\in\Omega_i$,
  $$\epsilon_i(x_i)=\sum_{y_i\in\Omega_i}|\epsilon_i(x_i,y_i)|\leq
  2d_{tv}(B_i^t,\pi_i).$$ Then,

  \begin{align*}
    d_{tv}(P^t,\pi)&=\max_{x\in \Omega} \frac{1}{2}\sum_{y\in \Omega}
    |P^t(x,y)-\pi(y)|\\
    &=\max_{x\in \Omega} \frac{1}{2}\sum_{y\in \Omega} \left|\prod_i
    B_i^t(x_i,y_i)-\prod_i\pi_i(y_i)\right|\\
    &=\max_{x\in \Omega} \frac{1}{2}\sum_{y\in \Omega} \left|\prod_i
    \left(B_i^t(x_i,y_i)-\pi_i(y_i) +
    \pi_i(y_i)\right)-\prod_i\pi_i(y_i)\right|\\
    &=\max_{x\in \Omega} \frac{1}{2}\sum_{y\in \Omega} \left|\prod_i
    \left(\epsilon_i(x_i,y_i) + \pi_i(y_i)\right)-\prod_i\pi_i(y_i)\right|\\
    &=\max_{x\in \Omega} \frac{1}{2}\sum_{y\in \Omega} \left|\sum_{S\subseteq
    [M], S\neq \emptyset}\prod_{i\in S}\epsilon_i(x_i,y_i)\prod_{i\notin
    S}\pi_i(y_i)\right|\\
    &\leq \max_{x\in \Omega} \frac{1}{2}\sum_{y\in \Omega} \sum_{S\subseteq [M],
    S\neq \emptyset}\prod_{i\in S}|\epsilon_i(x_i,y_i)|\prod_{i\notin
    S}|\pi_i(y_i)| \\
    &=\max_{x\in \Omega} \frac{1}{2} \sum_{S\subseteq [M], S\neq
    \emptyset}\prod_{i\in S}\sum_{y_i\in
    \Omega_i}|\epsilon_i(x_i,y_i)|\prod_{i\notin S}\sum_{y_i\in
    \Omega_i}|\pi_i(y_i)|\\
    &=\max_{x\in \Omega} \frac{1}{2} \sum_{S\subseteq [M], S\neq
    \emptyset}\prod_{i\in S}\epsilon_i(x_i)\prod_{i\notin S}1\\
    &=\max_{x\in\Omega} \frac{1}{2} \prod_i (1+\epsilon_i(x_i)) - 1/2\ \  \leq\
    \ \frac{1}{2} \prod_i (1+2d_{tv}(B_i^t,\pi_i)) - 1/2,
  \end{align*}

  as desired.  Thus in order to get $d_{tv}(P^t,\pi)\leq \epsilon,$ it suffices
  to show $d_{tv}(B_i^t,\pi_i)\leq \epsilon/(2M)$ for each $i$, because then
  \begin{align*}
    1+2d_{tv}(P^t,\pi)&\leq \prod_{i} (1+2d_{tv}(B_i^t,\pi_i))\\
    &\leq \prod_{i} (1+2\epsilon/(2M))\\
    &\leq e^{\epsilon} \ \leq \ 1+2\epsilon.
  \end{align*}
  Hence it suffices to show $d_{tv}(B_i^t,\pi_i)\leq \epsilon/(2M)$ for each $i$.

  Since $B_i^t= (p_iP_i+(1-p_i)I)^t = \sum_{j=0}^t \binom{t}{j}
  p_i^j(1-p_i)^{t-j}P_i^jI,$ we have
  \begin{align*}
    d_{tv}(B_i^t,\pi_i)&=\max_{x_i\in \Omega_i} \frac{1}{2}
    \sum_{y_i\in\Omega_i} \left|B_i^t(x_i,y_i)-\pi_i(y_i)\right|\\
    &=\max_{x_i\in \Omega_i} \frac{1}{2} \sum_{y_i\in\Omega_i}
    \left|\sum_{j=0}^t \binom{t}{j}p_i^j(1-p_i)^{t-j}
    P_i^j(x_i,y_i)-\pi_i(y_i)\right|\\
    &=\max_{x_i\in \Omega_i} \frac{1}{2} \sum_{y_i\in\Omega_i}
    \left|\sum_{j=0}^t \binom{t}{j}p_i^j(1-p_i)^{t-j}
    P_i^j(x_i,y_i)-\sum_{j=0}^t
    \binom{t}{j}p_i^j(1-p_i)^{t-j}\pi_i(y_i)\right|\\
    &\leq\max_{x_i\in \Omega_i} \frac{1}{2} \sum_{y_i\in\Omega_i} \sum_{j=0}^t
    \binom{t}{j}p_i^j(1-p_i)^{t-j} \left|P_i^j(x_i,y_i)-\pi_i(y_i)\right|\\
    &=\sum_{j=0}^t \binom{t}{j}p_i^j(1-p_i)^{t-j} \max_{x_i\in \Omega_i}
    \frac{1}{2} \sum_{y_i\in\Omega_i} \left|P_i^j(x_i,y_i)-\pi_i(y_i)\right|\\
    &=\sum_{j=0}^t \binom{t}{j}p_i^j(1-p_i)^{t-j} d_{tv}(P_i^j,\pi_i).
  \end{align*}
  Let $t_i= \tau_i(\epsilon/(4M))$.  Now, for $j\geq t_i=\tau_i(\epsilon/(4M))$,
  we have that $d_{tv}(P_i^j,\pi_i)<\epsilon/(4M)$.  For all $j$, we have
  $d_{tv}(P_i^j,\pi_i)\leq 2$, so if $X$ is a binomial random variable with
  parameters $t$ and $p_i$, we have

  \begin{align*}
    d_{tv}(B_i^t,\pi_i)&\leq\sum_{j=0}^t \binom{t}{j}p_i^j(1-p_i)^{t-j}
    d_{tv}(P_i^j,\pi_i)\\
    &=\sum_{j=0}^{t_i-1} \binom{t}{j}p_i^j(1-p_i)^{t-j} d_{tv}(P_i^j,\pi_i)+
    \sum_{j=t_i}^{t}\binom{t}{j}p_i^j(1-p_i)^{t-j} d_{tv}(P_i^j,\pi_i)\\
    &<2\sum_{j=0}^{t_i-1} \binom{t}{j}p_i^j(1-p_i)^{t-j} +
    \sum_{j=t_i}^{t}\binom{t}{j}p_i^j(1-p_i)^{t-j} \epsilon/(2M)\\
    &=2P(X<t_i) + \epsilon/(2M).
  \end{align*}
  By Chernoff bounds, $P(X<(1-\delta)tp_i)\leq e^{-tp_i \delta^2/2}.$ Setting
  $\delta=1-t_i/(tp_i)$, then for all $t>2t_i/p_i$, $\delta^2\geq 1/4$ and we
  have $$P(X<t_i)\leq e^{-tp_i \delta^2/2}\leq e^{-tp_i /8}\leq \epsilon/(8M),$$
  as long as $t\geq 8\ln(\epsilon/(8M))/p_i.$ Therefore for $t\geq \max\{
  8\ln(\epsilon/(8M))/p_i, 2t_i/p_i\},$
  \begin{align*}
    d_{tv}(B_i^t,\pi_i)&= 2P(X<t_i) + \epsilon/(4M) \\
    &\leq 2\epsilon/(8M) + \epsilon/(4M) \ = \ \epsilon/(2M).
  \end{align*}
  Hence by time $t$ the total variation distance satisfies $d_{tv}(P^t,\pi)\leq
  \epsilon.$
\end{proof}

\vspace{.2in}
\noindent
{\bf Acknowledgments}.
The authors are grateful to Jim Fill for introducing them to this problem and
for several useful conversations about the problem.

\newpage
\bibliographystyle{plain}
\bibliography{bib}
\end{document}